\documentclass[reqno]{amsart}
\usepackage[utf8]{inputenc}
\usepackage{pmboxdraw}
\usepackage{amsmath}
\usepackage{amssymb}
\usepackage{xcolor}
\usepackage{graphicx}
\usepackage{float}
\usepackage[hidelinks]{hyperref}

\calclayout

\usepackage[left=2.0cm,right=2.0cm,top=2.0cm,bottom=2.0cm]{geometry}

\usepackage{cite}
\newtheorem{theorem}{Theorem}
\newtheorem{remark}[theorem]{Remark}

\newtheorem{proposition}[theorem]{Proposition}

\newtheorem{corollary}[theorem]{Corollary}

\newtheorem{problem}{Problem}

\newcommand{\Tr}{\mathrm{Tr}}
\newcommand{\C}{\mathbb{C}}

\newcommand{\R}{\mathbb{R}}

\newcommand{\ler}[1]{\left( #1 \right)}
\newcommand{\lesq}[1]{\left[ #1 \right]}
\newcommand{\lers}[1]{\left\{ #1 \right\}}
\newcommand{\hohc}{\cH \otimes \cH^*}

\newcommand{\abs}[1]{\left| #1 \right|}

\newcommand{\norm}[1]{\left|\left|#1\right|\right|}
\newcommand{\bra}[1]{\langle #1 |}
\newcommand{\ket}[1]{| #1 \rangle}
\newcommand{\Bra}[1]{\langle\langle #1 ||}
\newcommand{\Ket}[1]{|| #1 \rangle\rangle}

\newcommand{\be}{\begin{equation}}
\newcommand{\ee}{\end{equation}}
\newcommand{\ba}{\begin{array}}
\newcommand{\ea}{\end{array}}

\newcommand{\cH}{\mathcal{H}}
\newcommand{\cB}{\mathcal{B}}

\newcommand{\cP}{\mathcal{P}}

\newcommand{\cS}{\mathcal{S}}
\newcommand{\cC}{\mathcal{C}}

\newcommand{\cT}{\mathcal{T}}
\newcommand{\cW}{\mathcal{W}}
\newcommand{\cX}{\mathcal{X}}
\newcommand{\cL}{\mathcal{L}}

\newcommand{\tr}{\mathrm{tr}}
\newcommand{\dd}{\mathrm{d}}

\newcommand{\cA}{\mathcal{A}}

\newcommand\lh{\cL(\cH)}

\newcommand{\sh}{\cS\ler{\cH}}

\title[Kantorovich duality and optimal couplings of qubits]{Strong Kantorovich duality for quantum optimal transport with generic cost and optimal couplings on quantum bits}

\author[Gergely Bunth]{Gergely Bunth}
\address{Gergely Bunth, HUN-REN Alfr\'ed R\'enyi Institute of Mathematics\\ Re\'altanoda u. 13-15.\\Budapest H-1053\\ Hungary\\ and Department of Analysis and Operations Research, Institute of Mathematics, Budapest University of Technology and Economics\\ M\H{u}egyetem rkp. 3. \\ Budapest H-1111 \\ Hungary}
\email{bunth.gergely@renyi.hu}

\author[J\'ozsef Pitrik]{J\'ozsef Pitrik}
\address{J\'ozsef Pitrik, HUN-REN Wigner Research Centre for Physics\\ Budapest H-1525, Hungary\\ and HUN-REN Alfr\'ed R\'enyi Institute of Mathematics\\ Re\'altanoda u. 13-15.\\ Budapest H-1053\\ Hungary\\ and Department of Analysis and Operations Research, Institute of Mathematics \\Budapest University of Technology and Economics\\ M\H{u}egyetem rkp. 3. \\ Budapest H-1111\\ Hungary}
\email{pitrik.jozsef@renyi.hu}

\author[Tam\'as Titkos]{Tam\'as Titkos}
\address{Tam\'as Titkos, Corvinus University of Budapest\\ Department of Mathematics\\ Fővám tér 13-15.\\ Budapest H-1093\\Hungary\\ 
and \\ HUN-REN Alfr\'ed R\'enyi Institute of Mathematics\\ Re\'altanoda u. 13-15.\\ Budapest H-1053\\Hungary}
\email{tamas.titkos@uni-corvinus.hu}

\author[D\'aniel Virosztek]{D\'aniel Virosztek}
\address{D\'aniel Virosztek, HUN-REN Alfr\'ed R\'enyi Institute of Mathematics\\ Re\'altanoda u. 13-15.\\Budapest H-1053\\ Hungary}
\email{virosztek.daniel@renyi.hu}

\date{}

\subjclass[2020]{Primary: 49Q22; 81P16. Secondary: 81Q10.}

\keywords{quantum optimal transport, generic cost, quantum channels}

\thanks{Bunth was supported by the Momentum Program of the Hungarian Academy of Sciences (grant no. LP2021-15/2021); Pitrik was supported by the “Frontline” Research Excellence Programme of the Hungarian National Research, Development and Innovation Office - NKFIH (grant no. KKP133827) and by the Momentum Program of the Hungarian Academy of Sciences (grant no. LP2021-15/2021);T. Titkos is supported by the Hungarian National Research, Development and Innovation Office (NKFIH) under grant agreements no. K134944 and no. Excellence\_151232, and by the Momentum program of the Hungarian Academy of Sciences under grant agreement no. LP2021-15/2021. D. Virosztek is supported by the Momentum program of the Hungarian Academy of Sciences under grant agreement no. LP2021-15/2021, by the Hungarian National Research, Development and Innovation Office (NKFIH) under grant agreement no. Excellence\_151232, and partially supported by the ERC Synergy Grant No. 810115.}

\begin{document}

\begin{abstract}
We prove Kantorovich duality for a linearized version of a recently proposed non-quadratic quantum optimal transport problem, where quantum channels realize the transport. As an application, we determine optimal solutions of both the primal and the dual problem using this duality in the case of quantum bits and distinguished cost operators, with certain restrictions on the states involved. Finally, keeping the same restrictions regarding the states involved, we use this information on optimal solutions to give an analytical proof of the triangle inequality even for the square of the induced quantum Wasserstein divergences.
\end{abstract}

\maketitle

\section{Introduction}

\subsection{Motivation and main result}

Although Gaspard Monge formulated the first version of the optimal transport problem already at the end of the 18th century \cite{monge1781memoire}, the theory of optimal transportation became a vital part of mathematical analysis only in the 20th century, when major advances were obtained by Leonid Kantorovich in the 1940s \cite{kantorovich1942-transloc,kantorovich1948-monge}, and the breakthrough result of Yann Brenier on the structure of optimal transport maps \cite{Brenier-polar-fra, Brenier-polar-en} induced intense research activity on the topic by various authors working in analysis and mathematical physics.
\par
Techniques relying on the theory of optimal transport and using desirable properties of the induced Wasserstein distances on probability measures played a key role in significant advancements in several areas of mathematics, such as probability theory \cite{bgl,Butkovsky}, the study of physical evolution equations \cite{JKO-97a,JKO-97b,JKO-98} and stochastic partial differential equations \cite{Hairer2,Hairer3}, variational analysis \cite{FigalliMaggi1,FigalliMaggi2}, and the geometry of metric measure spaces \cite{LottVillani,Sturm4,Sturm5,Sturm6}.
We refer to the monographs \cite{Ambrosiobook, Figalli-book, Villani1, Villani2} for a detailed overview of the field.
\par
Beyond their theoretical importance, transport-related metrics and optimal transport techniques have found their place in a large variety of disciplines outside mathematics, such as economics \cite{Galichon}, finance, and biology \cite{SCHIEBINGER}, and they also became popular and found their applications in applied sciences like biomedical image processing \cite{bia1,bia2,images1}, data analysis and classification \cite{dia1,imageprocessing1}, or machine learning \cite{MC1, SK1,PC,m2,MachineLearning3}.

\par
Recent decades have seen also several non-commutative (or quantum) versions of the optimal transport problem and induced Wasserstein distances. In the early 1990s, relying on duality phenomena, Connes and Lott proposed a spectral distance in the framework of non-commutative geometry \cite{Connes-Lott}. A few years later, S{\l}omczy\'{n}ski and \.Zyczkowski defined a distance on quantum states by the classical Wasserstein distance of their Husumi transforms \cite{ZyczkowskiSlomczynski1,ZyczkowskiSlomczynski2}, and a free probability approach was proposed by Biane and Voiculescu in 2001 \cite{BianeVoiculescu} --- see also the works of Shlyakhtenko \cite{Shlyakhtenko-free-Wass,Shlyakhtenko-free-monotone-transport,Shlyakhtenko-free-transport} on the topic. Carlen and Maas laid down the foundations of a dynamical theory \cite{CarlenMaas-1,CarlenMaas-2,CarlenMaas-3, CarlenMaas-4} relying on the classical Benamou-Brenier formula and Jordan-Kinderlehrer-Otto theory, and this work has been continued by Datta, Rouz\'e \cite{DattaRouze1, DattaRouze2}, and Wirth \cite{Wirth-dual}, among others. Caglioti, Golse, Mouhot, and Paul worked out a quantum optimal transport concept based on quantum couplings \cite{CagliotiGolsePaul, CagliotiGolsePaul-towardsqot, GolseMouhotPaul, GolsePaul-Schrodinger, GolsePaul-wavepackets, GolsePaul-meanfieldlimit, GolseTPaul-pseudometrics,GolsePaul-OTapproach}, while De Palma and Trevisan established a similar, yet different, concept based on quantum channels \cite{DPT-AHP, DPT-lecture-notes}. The concept of Friedland, Eckstein, Cole and Życzkowski \cite{Friedland-Eckstein-Cole-Zyczkowski-MK, Cole-Eckstein-Friedland-Zyczkowski-QOT, BistronEcksteinZyczkowski} is also based on couplings, but with strikingly different cost operators. Duvenhage used modular couplings to define quantum Wasserstein distances \cite{Duvenhage1, Duvenhage-quad-Wass-vNA, Duvenhage3, Duvenhage-ext-quantum-det-bal}, and separable quantum Wasserstein distances have also been introduced and studied \cite{TothPitrik, toth-pitrik-2025, beatty-franca}. A substantial part of the above mentioned current approaches to non-commutative optimal transport is covered by the book \cite{OTQS-book}, and the reader is advised to consult the survey papers \cite{beatty2025wasserstein} and \cite{Trevisan-review-QOT} as well. 

\par 

In this paper, we take the quantum optimal transport concept developed by De Palma and Trevisan \cite{DPT-AHP, DPT-lecture-notes} as starting point, and consider a non-quadratic generalization of the transport problem introduced there, which we proposed recently in \cite{BPTV-p-Wass}. We consider a linear relaxation of this latter transport problem and prove strong Kantorovich duality for it with an appropriate dual problem. When proving the duality, we will follow the approach of Caglioti, Golse, and Paul \cite{CagliotiGolsePaul-towardsqot, GolseTPaul-pseudometrics} (see also \cite[Section 4.3]{Golse-lecture-notes}), which partially relies on ideas from the proof of the classical Kantorovich duality (see, e.g., \cite[Theorem 1.3]{Villani1}).
As an application, we determine optimal solutions of both the primal and the dual problem using this duality in the case of quantum bits and distinguished cost operators, with certain restrictions on the states involved. Finally, keeping the same restrictions regarding the states involved, we use this information on optimal solutions to give an analytical proof of the triangle inequality even for the square of the induced quantum Wasserstein divergences.

\subsection{Basic notions, notation}

Let us recall now those elements of the mathematical formalism of quantum mechanics that we will use throughout this paper. Let $\cH$ be a separable complex Hilbert space. In the sequel, we denote by $\lh^{sa}$ the set of self-adjoint but not necessarily bounded operators on $\cH$, and $\cS(\cH)$ stands for the set of states, that is, the set of positive trace-class operators on $\cH$ with unit trace. The space of all bounded operators on $\cH$ is denoted by $\cB(\cH),$ and we recall that the collection of trace-class operators on $\cH$ is denoted by $\mathcal{T}_1(\cH)$ and defined by $\mathcal{T}_1(\cH)= \lers{X \in \cB(\cH) \, \middle| \, \tr_{\cH}[\sqrt{X^*X}] < \infty}.$ Similarly, $\mathcal{T}_2(\cH)$ stands for the set of Hilbert-Schmidt operators defined by $\mathcal{T}_2(\cH)= \lers{X \in \cB(\cH) \, \middle| \, \tr_{\cH}[X^*X] < \infty}.$ A \emph{quantum channel} is a completely positive and trace preserving (CPTP) linear map on $\cT_1(\cH).$ The \emph{transpose} $A^T$ of a linear operator $A$ acting on a Hilbert space $\cH$ is a linear operator on the dual space $\cH^*$ defined by the identity $(A^T \eta) (\varphi) \equiv \eta (A \varphi)$ where $\eta \in \cH^*$ and $\varphi \in \cH.$
\par
We briefly recall also the classical optimal transport problem. If $\mu$ and $\nu$ are Borel probability measures on a complete and separable metric space $(\cX,d)$ representing the capacity of production and intensity of consumption of the goods to be transported, respectively, and $c: \cX \times \cX \rightarrow \R$ is a non-negative lower semicontinuous function representing the transport cost in the sense that $c(x,y)$ is the cost of transporting one unit of goods from $x$ to $y,$ then finding the optimal (that is, cheapest) transport plan is mathematically formalized as follows:
\be \label{eq:cl-ot-prob}
\text{minimize } \pi \mapsto \iint_{\cX \times \cX} c(x,y) \dd \pi(x,y)
\ee
where $\pi$ runs over all possible couplings of $\mu$ and $\nu.$ A measure $\pi \in \mathrm{Prob}(\cX \times \cX)$ is called a coupling of $\mu$ and $\nu$ (in notation: $\pi \in \cC(\mu,\nu)$) if the marginals of $\pi$ are $\mu$ and $\nu,$ that is,  
$\iint_{\cX \times \cX} f(x) \dd \pi(x,y)=\int_{\cX} f(x) \dd \mu(x)$ and $\iint_{\cX \times \cX} g(y) \dd \pi(x,y)=\int_{\cX} g(y) \dd \nu(y)$ for all continuous and bounded functions $f,g \in C_b(\cX).$ A consequence of the tightness (that is, sequential compactness in the weak topology) of $\cC(\mu, \nu)$ and the lower-semicontinuity of $c$ is that there is a coupling (in other words: transport plan) $\pi_0 \in \cC(\mu, \nu)$ that minimizes \eqref{eq:cl-ot-prob}, see, e.g., \cite[Thm. 4.1.]{Villani2}. If the cost function is the power of order $p$ of the distance, that is, $c(x,y)=d(x,y)^p,$ then optimal transport plans determine a genuine distance called $p$-Wasserstein distance and denoted by $d_{\cW_p}$ on probability measures:
\begin{align} \label{eq:classical-p-Wass-def}
d_{\cW_p}(\mu,\nu)=\ler{\inf_{\pi \in \cC(\mu,\nu)} \lers{\iint_{\cX \times \cX} d^p(x,y) \dd \pi(x,y)}}^{\frac{1}{p}}.    
\end{align}

An influential work of De Palma and Trevisan introduced a quantum mechanical counterpart of the classical optimal transport problem with quadratic cost, and also quadratic Wasserstein distances induced by optimal solutions of these transport problems \cite{DPT-AHP}. A key idea of this quantum optimal transport concept is that the transport between quantum states is realized by quantum channels \cite{DPT-AHP,DPT-lecture-notes}. A brief summary of their approach reads as follows. The inputs of the transport problem are the initial and final states $\rho, \omega \in \sh,$ where $\cH$ is a separable Hilbert space, and a finite collection of observable quantities $\cA=\lers{A_1, \dots, A_K}$ where $A_k \in \cL(\cH)^{sa}$ for all $k.$ The transport plans between $\rho$ and $\omega$ are quantum channels $\Phi: \cT_1\ler{\mathrm{supp}(\rho)} \to \cT_1(\cH)$ sending $\rho$ to $\omega,$ and a transport plan $\Phi$ gives rise to the quantum coupling $\Pi_{\Phi}$ the following way: 

\begin{align} \label{eq:Pi-phi-def}
    \Pi_{\Phi}=\ler{\Phi \otimes \mathrm{id}_{\mathcal{T}_1\ler{\cH^*}}} \ler{\Ket{\sqrt{\rho}}\Bra{\sqrt{\rho}}},
\end{align}
where $\Ket{\sqrt{\rho}}\Bra{\sqrt{\rho}} \in \cS\ler{\hohc}$ is the canonical purification \cite{Holevo} of the state $\rho \in \sh.$  Here, and in the sequel, we use the canonical linear isomorphism between $\mathcal{T}_2(\cH)$ and $\hohc$ which is the linear extension of the map 
\begin{align} \label{eq:canon-isom}
    \psi \otimes \eta \mapsto \ket{\psi}\circ \eta \qquad \ler{\psi \in \cH, \, \eta \in \cH^* }. 
\end{align}
Accordingly, for an $X \in \mathcal{T}_2(\cH),$ the symbol $\Ket{X}$ denotes the map $\C \ni z \mapsto z X \in \mathcal{T}_2(\cH) \simeq \hohc,$ while $\Bra{X}$ stands for the map $\cT_2(\cH) \ni Y \mapsto \tr_{\cH}\lesq{X^* Y},$ where $X^*$ is the adjoint of $X.$ It is easy to check that $\Pi_{\Phi}$ defined in \eqref{eq:Pi-phi-def} is a state on $\hohc$ such that its first marginal is $\omega$ while the second marginal is $\rho^T,$ that is,
\begin{align}
    \tr_{\cH^*}\lesq{\Pi_{\Phi}}=\omega \text{ and } \tr_{\cH}\lesq{\Pi_{\Phi}}=\rho^T.
\end{align}
Therefore, the set of all quantum couplings of the states $\rho,\omega \in \sh$ (denoted by $\cC(\rho, \omega)$) was defined in \cite{DPT-AHP} by
\be \label{eq:q-coup-def}
\cC\ler{\rho, \omega}=\lers{\Pi \in \cS\ler{\cH \otimes \cH^*} \, \middle| \, \tr_{\cH^*} [\Pi]=\omega, \,  \tr_{\cH} [\Pi]=\rho^T}.
\ee
In other words, and this rephrasing will prove useful when formalizing the dual transport problems, a coupling of $\rho$ and $\omega$ is a state $\Pi$ on $\hohc$ such that 
\be \label{eq:part-trace-def}
\tr_{\hohc}[\ler{A\otimes I_{\cH}^T} \Pi]=\tr_{\cH} [\omega A] \text{ and }
\tr_{\hohc}\lesq{\ler{I_{\cH} \otimes B^{T}} \Pi}=\tr_{\cH^*} [\rho^T B^T]=\tr_{\cH} [\rho B]
\ee
for all bounded $A, B \in \cL(\cH)^{sa}.$ The analogy of the above definition of quantum couplings with the classical notion of couplings recapped below equation \eqref{eq:cl-ot-prob} is clear, and we note that $\cC\ler{\rho,\omega}$ is never empty, because the trivial coupling $\omega\otimes\rho^T$ belongs to $\cC\ler{\rho,\omega}$.
\par
The definition of couplings \eqref{eq:q-coup-def} proposed by De Palma and Trevisan \cite{DPT-AHP} is different from the definition proposed by Golse, Mouhot, Paul \cite{GolseMouhotPaul} in the sense that it involves the dual Hilbert space $\cH^*$ and hence the transpose operation. For a clarification of this difference, see Remark 1 in \cite{DPT-AHP} while for more detail on the latter concept of quantum couplings, the interested reader should consult \cite{CagliotiGolsePaul,CagliotiGolsePaul-towardsqot,GolseMouhotPaul,GolsePaul-Schrodinger,GolsePaul-wavepackets,GolsePaul-Nbody, GolsePaul-OTapproach, GolseTPaul-pseudometrics, GolsePaul-meanfieldlimit}.

\section{Kantorovich duality} \label{sec:Kantorovich-duality}

The goal of this section is to formalize a linear relaxation of a non-linear primal quantum optimal transport problem that we proposed in \cite{BPTV-p-Wass}, and to propose a corresponding dual problem for which we can prove strong Kantorovich duality following the approach of Caglioti, Golse, and Paul \cite{CagliotiGolsePaul-towardsqot,GolseTPaul-pseudometrics}, which is explained also in \cite[Sec. 4.3]{Golse-lecture-notes}.
\par
In \cite[Section 2.1]{BPTV-p-Wass} we considered the following quantum mechanical optimal transport problem: let $\cH$ be a separable Hilbert space, $\cA=\lers{A_1, \dots, A_K}$ a finite collection of observables on $\cH,$ and let $c: \R^K \times \R^K \to \R$ be a non-negative, lower semicontinuous classical cost function. The positive and possibly unbounded self-adjoint cost operator $C_{c}^{(\cA)}$ acting on a dense subspace of $\ler{\hohc}^{\otimes K}$ is defined by
\begin{align} \label{eq:C-c-A-def}
    C_{c}^{(\cA)}
    =\iint_{\R^K \times \R^K} c\ler{x_1, \dots, x_K, y_1, \dots, y_K} \dd E_1(y_1) \otimes \dd E_1^T (x_1) \otimes \cdots \otimes \dd E_K(y_K) \otimes \dd E_K^T (x_K), 
\end{align}
where $E_k$ is the spectral measure of $A_K,$ that is, $A_k=\int_\R \lambda \dd E_k(\lambda)$ for $k \in \lers{1, \dots, K}.$
The transport problem is to
\begin{align} \label{eq:IID-primal-problem}
    \text{minimize } \tr_{\ler{\hohc}^{\otimes K}}\lesq{\Pi^{\otimes K} C_{c}^{(\cA)}}
\end{align}
where $\Pi$ runs over the set of all couplings of $\rho, \omega \in \sh,$ that is,
\begin{align} \label{eq:IID-primal-constraint}
    \Pi \in \cC(\rho,\omega)=\lers{\Pi \in \cS\ler{\cH \otimes \cH^*} \, \middle| \, \tr_{\cH^*} [\Pi]=\omega, \,  \tr_{\cH} [\Pi]=\rho^T}.
\end{align}
\par
It is important to note that the loss function $\tr_{\ler{\hohc}^{\otimes K}}\lesq{\Pi^{\otimes K} C_{c}^{(\cA)}}$ in the primal problem \eqref{eq:IID-primal-problem} is non-linear in its variable $\Pi \in \cC\ler{\rho,\omega}.$ However, there is a natural linear relaxation which is described the following way.

\begin{problem} \label{prob:linear-primal}
Let the initial and final states $\rho, \omega \in \sh$ and the cost operator $C_c^{(\cA)}$ acting on $\ler{\hohc}^{\otimes K}$ and defined by \eqref{eq:C-c-A-def} be given. The optimization task is to
\begin{align} \label{eq:linear-primal-problem}
\text{minimize } 
\tr_{\ler{\hohc}^{\otimes K}}\lesq{\Gamma C_c^{\cA}}
\end{align}
subject to the constraints
\begin{align} \label{eq:linear-primal-constraint}
\Gamma \in \cS\ler{\ler{\hohc}^{\otimes K}} , \, \ler{\Gamma}_{2k-1}=\omega, \, \ler{\Gamma}_{2k}=\rho^T \text{ for all } k\in \lers{1, \dots, K},    
\end{align}
where 
\begin{align} \label{eq:Gamma-2k-1-def}
    \ler{\Gamma}_{2k-1}
    =\tr_{1,\dots, 2k-2,2k,\dots,2K}\lesq{\Gamma}
    =\tr_{\ler{\hohc}^{\otimes(k-1)} \otimes \cH^* \otimes \ler{\hohc}^{\otimes(K-k)}}\lesq{\Gamma},
\end{align}
and 
\begin{align} \label{eq:Gamma-2k-def}
    \ler{\Gamma}_{2k}
    =\tr_{1,\dots, 2k-1,2k+1,\dots,2K}\lesq{\Gamma}
    =\tr_{\ler{\hohc}^{\otimes(k-1)} \otimes \cH \otimes \ler{\hohc}^{\otimes(K-k)}}\lesq{\Gamma}.   
\end{align}
\end{problem}
Note that $\Pi^{\otimes K}$ satisfies the constraint \eqref{eq:linear-primal-constraint} whenever $\Pi \in \cC(\rho, \omega)$ (defined in \eqref{eq:q-coup-def}), and therefore, the infimum of \eqref{eq:IID-primal-problem} is lower bounded by the infimum of \eqref{eq:linear-primal-problem}. The difference between the non-linear problem \eqref{eq:IID-primal-problem} and its linear relaxation (Problem \ref{prob:linear-primal}) is that the couplings of $\rho$ and $\omega$ acting on different subsystems are required to be independent in the former version while they may have correlations in the latter version. We will present an explicit example in the sequel (see Proposition \ref{prop:strictly-smaller}) which demonstrates that minimum of \eqref{eq:linear-primal-problem} can be strictly smaller than that of \eqref{eq:IID-primal-problem}.
\par
It is instructive to consider the case when the transport cost factorizes, that is,
\begin{align} \label{eq:tr-cost-factorizes}
    c(x_1, \dots, x_K, y_1, \dots, y_K)=f_1\ler{x_1,y_1}+ \dots, +f_K\ler{x_K,y_K}.
\end{align}
In this case, the cost operator $C_c^{(\cA)}$ defined in \eqref{eq:C-c-A-def} has the simpler form
\begin{align}
    C_c^{(\cA)}=\sum_{k=1}^K
    I_{\hohc}^{\otimes(k-1)}
    \otimes
    \ler{\iint_{\R \times \R} f_k\ler{x_k,y_k} \dd  E_k(y_k) \otimes \dd E_k^T (x_k)}  
    \otimes
    I_{\hohc}^{\otimes(K-k)}.
\end{align}
Therefore, introducing the shorthand $C_k:=\iint_{\R \times \R} f_k\ler{x_k,y_k} \dd  E_k(y_k) \otimes \dd E_k^T (x_k),$ for any $\Gamma \in \cS\ler{\ler{\hohc}^{\otimes K}}$ one gets
\begin{align}
    \tr_{\ler{\hohc}^{\otimes K}}\lesq{\Gamma C_c^{\cA}}
    =\sum_{k=1}^K \tr_{\hohc}\lesq{\ler{\Gamma}_{(2k-1,2k)} C_k},
\end{align}
where the marginals $\ler{\Gamma}_{(2k-1,2k)}$ are defined similarly as in Problem \ref{prob:linear-primal}, that is, $$\ler{\Gamma}_{(2k-1,2k)}=\tr_{1,\dots, 2k-2, 2k+1, \dots, 2K}\lesq{\Gamma}.$$
Consequently, the linearized primal problem (Problem \ref{prob:linear-primal}) reduces to the following:
\begin{align} \label{eq:linear-factorized}
    \text{minimize } \sum_{k=1}^K \tr_{\hohc} \lesq{\Pi_k C_k}
\end{align}
under the constraints
\begin{align}
    \Pi_1, \dots, \Pi_K \in \cC(\rho, \omega).
\end{align}
On the contrary, the non-linear primal problem \eqref{eq:IID-primal-problem} proposed in \cite[Sec. 2.1]{BPTV-p-Wass} reduces to 
\begin{align} \label{eq:primal-factorized}
    \text{minimize }\tr_{\hohc} \lesq{\Pi \ler{\sum_{k=1}^K C_k}} \text{ subject to } \Pi \in \cC(\rho, \omega),
\end{align}
as noted in \cite{BPTV-p-Wass} for the special case $c(x_1, \dots, x_K, y_1, \dots, y_K)=\sum_{k=1}^p\abs{x_k-y_k}^p.$ Note that if the transport cost factorizes in the sense of \eqref{eq:tr-cost-factorizes}, then the a priori non-linear loss function of the primal problem \eqref{eq:IID-primal-problem} becomes linear, as clearly shown by \eqref{eq:primal-factorized}.
\par
The classical dual problem for the optimal transportation problem \eqref{eq:cl-ot-prob} on the complete and separable metric space $\cX$ is to
\begin{align} \label{eq:cl-dual-prob}
    \text{maximize } \int_{\cX} \psi(y) \dd \nu(y) +\int_{\cX} \varphi(x) \dd \mu(x)
\end{align}
subject to the constraint
\begin{align} \label{eq:cl-dual-const}
    \psi(y)+\varphi(x) \leq c(x,y)
\end{align}
for all $x,y \in \cX,$
and the classical Kantorovich duality asserts that
\begin{align} \label{eq:class-Kant-dual}
    \sup \lers{ \int_{\cX} \psi(y) \dd \nu(y) +\int_{\cX} \varphi(x) \dd \mu(x) \, \middle| \, \psi(y)+\varphi(x) \leq c(x,y)}= \nonumber \\
    = \min \lers{ \iint_{\cX \times \cX} c(x,y) \dd \pi(x,y) \, \middle| \, \pi \in \cC(\mu, \nu)},
\end{align}
see, e.g., Theorem 1.3. in \cite{Villani1}.
In view of \eqref{eq:cl-dual-prob} and \eqref{eq:cl-dual-const}, a natural generalization of the classical dual problem to our quantum setting is the following.

\begin{problem} \label{prob:linear-dual}
Let the initial and final states $\rho, \omega \in \sh$ and the cost operator $C_c^{(\cA)}$ defined by \eqref{eq:C-c-A-def} and acting on $\ler{\hohc}^{\otimes K}$ be given. The optimization task is to
\begin{align} \label{eq:linear-k-partite-dual-problem}
    \text{maximize } \sum_{k=1}^K \ler{\tr_{\cH}\lesq{\omega Y_k} + \tr_{\cH}\lesq{\rho X_k}}
\end{align}
subject to the constraint
\begin{align} \label{eq:linear-k-partite-dual-constraint}
    \sum_{k=1}^K I_{\hohc}^{\otimes (k-1)} \otimes \ler{Y_k \otimes I_{\cH^*}+I_{\cH} \otimes X_k^T } \otimes I_{\hohc}^{\otimes (K-k)} \leq C_c^{\cA}.
\end{align}    
\end{problem}

It turns out that the above proposed Problem \ref{prob:linear-dual} is indeed the strong Kantorovich dual of Problem \ref{prob:linear-primal}. The precise statement is formalized in the following theorem, which is the main result of this section.

\begin{theorem} \label{thm:linear-duality}
Let $\cA=\lers{A_1,\dots,A_K}$ be a finite collection of observables on a separable Hilbert space $\cH,$ let the cost operator $C_c^{(\cA)}$ be defined as in \eqref{eq:C-c-A-def}, and let $\rho, \omega \in \sh.$ Then
    \begin{align} \label{eq:duality-nr-1}
        \sup \lers{\sum_{k=1}^K \ler{\tr_{\cH}\lesq{\omega Y_k} + \tr_{\cH}\lesq{\rho X_k}} \, \middle| \, \sum_{k=1}^K I_{\hohc}^{\otimes (k-1)} \otimes \ler{Y_k \otimes I_{\cH^*}+I_{\cH} \otimes X_k^T } \otimes I_{\hohc}^{\otimes (K-k)} \leq C_c^{\cA}} = \nonumber \\
        =\min \lers{\tr_{\ler{\hohc}^{\otimes K}}\lesq{\Gamma C_c^{(\cA)}} \, \middle| \, \Gamma \in \cS\ler{\ler{\hohc}^{\otimes K}} , \, \ler{\Gamma}_{2k-1}=\omega, \, \ler{\Gamma}_{2k}=\rho^T \text{ for all } k\in \lers{1, \dots, K}},
    \end{align}
    where the variables $X_1, Y_1, \dots, X_K, Y_K$ to be optimized are self-adjoint and bounded operators on $\cH,$ and the marginals $\ler{\Gamma}_{2k-1}$ and $\ler{\Gamma}_{2k}$ are the ones defined in \eqref{eq:Gamma-2k-1-def} and \eqref{eq:Gamma-2k-def}.
\end{theorem}

\begin{proof}
Let us define the functional $\Theta: \cB\ler{\ler{\hohc}^{\otimes K}}^{sa} \rightarrow (-\infty,+\infty]$ by
\begin{align} \label{eq:IID-Theta-def}
    \Theta(U):=\begin{cases}
        0 & \text{if } U \geq -C_c^{\cA} \\
        + \infty & \text{else.}
    \end{cases}
\end{align}
Here, the inequality $U \geq -C_c^{\cA}$ is to be understood in the L\"owner sense, that is, $\bra{x}U\ket{x} \geq -\bra{x}C_c^{\cA}\ket{x}$ for all $x \in \mathrm{dom}\ler{C_c^{\cA}}.$ The constraint $U \geq -C_c^{\cA}$ defines a convex domain in $\cB\ler{\ler{\hohc}^{\otimes K}}^{sa},$ and hence the functional $\Theta$ defined by \eqref{eq:IID-Theta-def} is convex. Furthermore, we define the functional $\Xi: \cB\ler{\ler{\hohc}^{\otimes K}}^{sa} \rightarrow (-\infty,+\infty]$ by
\begin{align} \label{eq:IID-Xi-def}
    \Xi(U):=\begin{cases}
        \sum_{k=1}^K \ler{\tr_{\cH}\lesq{\omega Y_k} + \tr_{\cH}\lesq{\rho X_k}} & \text{if } U= \sum_{k=1}^K I_{\hohc}^{\otimes (k-1)} \otimes \ler{Y_k \otimes I_{\cH^*}+I_{\cH} \otimes X_k^T } \otimes I_{\hohc}^{\otimes (K-k)} \\
        + \infty & \text{else,}
    \end{cases}
\end{align}
where $X_1,Y_1, \dots,X_K, Y_K \in \cB(\cH)^{sa}.$ It is important to note that the domain of $\Xi,$ that is, the region where it takes finite values, is convex. Indeed, it is a direct sum of linear subspaces of $\cB\ler{\ler{\hohc}^{\otimes K}}^{sa},$ namely, 
\begin{align} \label{eq:IID-Xi-domain-char}
\mathrm{domain}(\Xi)=\bigoplus_{k=1}^K\ler{
I_{\hohc}^{\otimes (k-1)} \otimes \cB(\cH)^{sa}\otimes I_{\cH}^T\otimes I_{\hohc}^{\otimes (K-k)} \oplus
I_{\hohc}^{\otimes (k-1)} \otimes I_{\cH} \otimes \cB\ler{\cH^*}^{sa}\otimes I_{\hohc}^{\otimes (K-k)}}.
\end{align}
 Recall that the Legendre-Fenchel transform $\Omega^*$ of a convex function $\Omega$ defined on the real normed vector space $\cB\ler{\ler{\hohc}^{\otimes K}}^{sa}$ equipped with the operator norm topology is defined by
\begin{align} \label{eq:IID-LF-transform-def}
\Omega^*\ler{\widetilde{\Gamma}}:=\sup \lers{\widetilde{\Gamma}(U)-\Omega(U) \, \middle| \, U \in \cB\ler{\ler{\hohc}^{\otimes K}}^{sa}}
\end{align}
for all $\widetilde{\Gamma} \in \ler{\cB\ler{\ler{\hohc}^{\otimes K}}^{sa}}^*.$ The famous Fenchel-Rockafellar duality theorem \cite[Thm. 1.12]{Brezis-FA-book} asserts that 
\begin{align} \label{eq:IID-LF-dualiy}
    \inf\lers{\Theta(U)+\Xi(U) \, \middle| \, U \in \cB\ler{\ler{\hohc}^{\otimes K}}^{sa}}
    =\max \lers{-\Theta^*\ler{-\widetilde{\Gamma}}-\Xi^*\ler{\widetilde{\Gamma}} \, \middle| \, \widetilde{\Gamma} \in \ler{\cB\ler{\ler{\hohc}^{\otimes K}}^{sa}}^*}
\end{align}
whenever there exists a $U_0 \in \cB\ler{\ler{\hohc}^{\otimes K}}^{sa}$ such that both $\Theta(U_0)$ and $\Xi(U_0)$ are finite, and $\Theta$ is continuous at $U_0.$ Clearly, $U_0:=I_{\ler{\hohc}^{\otimes K}}$ does the job. 
Indeed, by \eqref{eq:IID-Theta-def}, we have $\Theta\ler{I_{\ler{\hohc}^{\otimes K}}}=0,$ and by \eqref{eq:IID-Xi-def} we get $\Xi\ler{I_{\ler{\hohc}^{\otimes K}}}=1.$ Moreover, $I_{\ler{\hohc}^{\otimes K}}$ lies in the interior of the cone of positive semidefinite operators in the operator norm topology on $\cB\ler{\ler{\hohc}^{\otimes K}}^{sa}.$ Recall that the classical cost function $c$ is nonnegative, and hence the induced cost operator $C_c^{\cA}$ defined in \eqref{eq:C-c-A-def} is positive semidefinite. Therefore, $U\geq -C_c^{\cA}$ holds for any positive semidefinite $U \in \cB(\hohc)^{sa},$ and hence there is an open neighborhood of $I_{\ler{\hohc}^{\otimes K}}$ where $\Theta$ vanishes. Consequently, $\Theta$ is continuous in $I_{\ler{\hohc}^{\otimes K}}.$
\par
By the definition of the Legendre-Fenchel transform \eqref{eq:IID-LF-transform-def} and straightforward steps, we can compute $\Theta^*$ as follows:
\begin{align} \label{eq:IID-Theta-star-computation}
\Theta^*\ler{-\widetilde{\Gamma}} = \sup \lers{-\widetilde{\Gamma}(U)- \Theta(U) \middle| \, U \in \cB\ler{\ler{\hohc}^{\otimes K}}^{sa}} = \sup_{U \, : \, U \geq -C_c^{\cA}} \lers{-\widetilde{\Gamma}(U)}=
\\ 
=-\inf_{U \, : \, U \geq -C_c^{\cA}} \lers{\widetilde{\Gamma}(U)}
=\begin{cases}
\widetilde{\Gamma}\ler{C_c^{\cA}}, & \text{if } \widetilde{\Gamma} \geq 0, \\
+\infty, & \text{else.}
\end{cases}
\end{align}
Indeed, if $\widetilde{\Gamma} \geq 0,$ that is, $\widetilde{\Gamma}(R)\geq 0$ for all positive semi-definite $R \in \cB\ler{\ler{\hohc}^{\otimes K}}^{sa},$ then
\begin{align}
    -\inf_{U \, : \, U \geq -C_c^{\cA}} \lers{\widetilde{\Gamma}(U)}
    =-\inf_{ S \geq 0, \, S \in \cB\ler{\ler{\hohc}^{\otimes K}}^{sa}} \lers{\widetilde{\Gamma}\ler{-C_c^{(\cA)}}+\widetilde{\Gamma}(S)}=-\widetilde{\Gamma}\ler{-C_c^{\cA}}.
\end{align}
On the other hand, if $\widetilde{\Gamma} \ngeq 0,$ that is, $\widetilde{\Gamma}(R)<0$ for some $R\geq 0,$ then $\widetilde{\Gamma}\ler{-C_c^{\cA}+tR}$ tends to $-\infty$ as $t$ tends to $+\infty,$ and hence $\inf_{U \, : \, U \geq -C_c^{\cA}}\lers{\widetilde{\Gamma}(U)}=-\infty.$
As for the convex conjugate of $\Xi,$ one gets
\begin{align}
\Xi^*\ler{\widetilde{\Gamma}} = \sup \lers{\widetilde{\Gamma}(U) - \Xi(U) \, \middle| \, U \in \cB\ler{\ler{\hohc}^{\otimes K}}^{sa}} = \nonumber \\
=\sup \lers{\widetilde{\Gamma}(U)-\ler{\sum_{k=1}^K \ler{\tr_{\cH}\lesq{\omega Y_k} + \tr_{\cH}\lesq{\rho X_k}}} \, \middle| \, U= \sum_{k=1}^K I_{\hohc}^{\otimes (k-1)} \otimes \ler{Y_k \otimes I_{\cH^*}+I_{\cH} \otimes X_k^T } \otimes I_{\hohc}^{\otimes (K-k)} } \nonumber
\end{align}
\begin{align}
=\sup_{X_1,Y_1,\dots,X_K,Y_k \in \cB(\cH)^{sa}}\lers{\widetilde{\Gamma}\ler{\sum_{k=1}^K I_{\hohc}^{\otimes (k-1)} \otimes \ler{Y_k \otimes I_{\cH^*}+I_{\cH} \otimes X_k^T } \otimes I_{\hohc}^{\otimes (K-k)}}-\ler{\sum_{k=1}^K  \ler{\tr_{\cH}\lesq{\omega Y_k} + \tr_{\cH}\lesq{\rho X_k}}}} \nonumber
\end{align}
\begin{align}  \label{eq:IID-Xi-star-computation}
=\begin{cases}
0, & \text{if }  \widetilde{\Gamma}\ler{\sum_{k=1}^K I_{\hohc}^{\otimes (k-1)} \otimes \ler{Y_k \otimes I_{\cH^*}+I_{\cH} \otimes X_k^T } \otimes I_{\hohc}^{\otimes (K-k)}}=\sum_{k=1}^K \ler{\tr_{\cH}\lesq{\omega Y_k} + \tr_{\cH}\lesq{\rho X_k}} \\
+\infty, & \text{else,}
\end{cases}    
\end{align}
where the condition in the first line of \eqref{eq:IID-Xi-star-computation} means that the equation holds for all $X_1, Y_1, \dots, X_k, Y_K \in \cB(\cH)^{sa}.$

On one hand, the left-hand side of \eqref{eq:IID-LF-dualiy} can be written as
\begin{align} \label{eq:IID-LHS-of-LF-duality}
\inf\lers{\Theta(U)+\Xi(U) \, \middle| \, U \in \cB\ler{\ler{\hohc}^{\otimes K}}^{sa}} \nonumber \\
=\inf\lers{\sum_{k=1}^K \ler{\tr_{\cH}\lesq{\omega Y_k} + \tr_{\cH}\lesq{\rho X_k}} \, \middle| \, \sum_{k=1}^K I_{\hohc}^{\otimes (k-1)} \otimes \ler{Y_k \otimes I_{\cH^*}+I_{\cH} \otimes X_k^T } \otimes I_{\hohc}^{\otimes (K-k)} \geq -C_c^{\cA} } \nonumber
\end{align}
\begin{align}
=\inf\lers{-\sum_{k=1}^K \ler{\tr_{\cH}\lesq{\omega(-Y_k)} + \tr_{\cH}\lesq{\rho( -X_k)}} \, \middle| \, \sum_{k=1}^K I_{\hohc}^{\otimes (k-1)} \otimes \ler{(-Y_k) \otimes I_{\cH^*}+I_{\cH} \otimes (-X_k)^T } \otimes I_{\hohc}^{\otimes (K-k)} \leq C_c^{\cA} } \nonumber
\end{align}
\begin{align}
=\inf\lers{-\sum_{k=1}^K \ler{\tr_{\cH}\lesq{\omega Y_k} + \tr_{\cH}\lesq{\rho X_k}} \, \middle| \, \sum_{k=1}^K I_{\hohc}^{\otimes (k-1)} \otimes \ler{Y_k \otimes I_{\cH^*}+I_{\cH} \otimes X_k^T } \otimes I_{\hohc}^{\otimes (K-k)} \leq C_c^{\cA} } \nonumber
\\
-\sup\lers{\sum_{k=1}^K \ler{\tr_{\cH}\lesq{\omega Y_k} + \tr_{\cH}\lesq{\rho X_k}} \, \middle| \, \sum_{k=1}^K I_{\hohc}^{\otimes (k-1)} \otimes \ler{Y_k \otimes I_{\cH^*}+I_{\cH} \otimes X_k^T } \otimes I_{\hohc}^{\otimes (K-k)} \leq C_c^{\cA} }.
\end{align}
On the other hand, by \eqref{eq:IID-Theta-star-computation} and \eqref{eq:IID-Xi-star-computation}, the right-hand side of \eqref{eq:IID-LF-dualiy} reads as
\begin{align} \label{eq:IID-RHS-of-LF-duality}
\max \lers{-\Theta^*\ler{-\widetilde{\Gamma}}-\Xi^*\ler{\widetilde{\Gamma}} \, \middle| \, \widetilde{\Gamma} \in \ler{\cB\ler{\ler{\hohc}^{\otimes K}}^{sa}}^*}
\end{align}
\begin{align} \label{eq:IID-RHS-of-LF-duality-2}
&=\max \lers{-\widetilde{\Gamma}(C_c^{\cA}) \, \middle| \, \widetilde{\Gamma} \geq 0,  \widetilde{\Gamma}\ler{\sum_{k=1}^K I_{\hohc}^{\otimes (k-1)} \otimes \ler{Y_k \otimes I_{\cH^*}+I_{\cH} \otimes X_k^T } \otimes I_{\hohc}^{\otimes (K-k)}}=\sum_{k=1}^K \ler{\tr_{\cH}\lesq{\omega Y_k} + \tr_{\cH}\lesq{\rho X_k}}} \nonumber
\\ 
&=-\min \lers{\widetilde{\Gamma}(C_c^{\cA}) \, \middle| \, \widetilde{\Gamma} \geq 0,  \widetilde{\Gamma}\ler{\sum_{k=1}^K I_{\hohc}^{\otimes (k-1)} \otimes \ler{Y_k \otimes I_{\cH^*}+I_{\cH} \otimes X_k^T } \otimes I_{\hohc}^{\otimes (K-k)}}=\sum_{k=1}^K \ler{\tr_{\cH}\lesq{\omega Y_k} + \tr_{\cH}\lesq{\rho X_k}}}.
\end{align}
Consequently, 
\begin{align} \label{eq:IID-KD-final}
\sup\lers{\sum_{k=1}^K \ler{\tr_{\cH}\lesq{\omega Y_k} + \tr_{\cH}\lesq{\rho X_k}} \, \middle| \, \sum_{k=1}^K I_{\hohc}^{\otimes (k-1)} \otimes \ler{Y_k \otimes I_{\cH^*}+I_{\cH} \otimes X_k^T } \otimes I_{\hohc}^{\otimes (K-k)} \leq C_c^{\cA} }
= \nonumber \\
=\min \lers{\widetilde{\Gamma}(C_c^{\cA}) \, \middle| \, \widetilde{\Gamma} \geq 0,  \widetilde{\Gamma}\ler{\sum_{k=1}^K I_{\hohc}^{\otimes (k-1)} \otimes \ler{Y_k \otimes I_{\cH^*}+I_{\cH} \otimes X_k^T } \otimes I_{\hohc}^{\otimes (K-k)}}=\sum_{k=1}^K \ler{\tr_{\cH}\lesq{\omega Y_k} + \tr_{\cH}\lesq{\rho X_k}}}.
\end{align}
It can be shown very similarly to the proof of \cite[Lemma 3.3]{CagliotiGolsePaul-towardsqot} that for any functional $\widetilde{\Gamma}$ satisfying the conditions described on the right-hand side of \eqref{eq:IID-KD-final} there exists a positive trace-class operator $\Gamma \in \cT_1\ler{\ler{\hohc}^{\otimes K}}$ such that $\widetilde{\Gamma}(U)=\tr_{\ler{\hohc}^{\otimes K}}\lesq{\Gamma U}$ for all $U \in \cB\ler{\ler{\hohc}^{\otimes K}}^{sa}.$ The requirement that 

\begin{align} \label{eq:partial-trace-requirement}
    \tr_{\ler{\hohc}^{\otimes K}}\lesq{\Gamma \ler{\sum_{k=1}^K I_{\hohc}^{\otimes (k-1)} \otimes \ler{Y_k \otimes I_{\cH^*}+I_{\cH} \otimes X_k^T } \otimes I_{\hohc}^{\otimes (K-k)}}}=\sum_{k=1}^K \ler{\tr_{\cH}\lesq{\omega Y_k} + \tr_{\cH}\lesq{\rho X_k}}
\end{align}
holds for all $X_1, \dots, X_K, Y_1, \dots, Y_K \in \cB(\cH)^{sa}$ is clearly equivalent to the condition
\begin{align}
    \ler{\Gamma}_{2k-1}=\omega, \, \ler{\Gamma}_{2k}=\rho^T \text{ for all } k\in \lers{1, \dots, K}, 
\end{align}
and hence \eqref{eq:IID-KD-final} can be written as

\begin{align} \label{eq:IID-KD-final-final}
\sup\lers{\sum_{k=1}^K \ler{\tr_{\cH}\lesq{\omega Y_k} + \tr_{\cH}\lesq{\rho X_k}} \, \middle| \, \sum_{k=1}^K I_{\hohc}^{\otimes (k-1)} \otimes \ler{Y_k \otimes I_{\cH^*}+I_{\cH} \otimes X_k^T } \otimes I_{\hohc}^{\otimes (K-k)} \leq C_c^{\cA} }
= \nonumber \\
=\min \lers{\tr_{\ler{\hohc}^{\otimes K}}\lesq{\Gamma C_c^{(\cA)}} \, \middle| \, \Gamma \geq 0,   \,
\ler{\Gamma}_{2k-1}=\omega, \, \ler{\Gamma}_{2k}=\rho^T \text{ for all } k\in \lers{1, \dots, K} 
},
\end{align}
as desired. 
\end{proof}

We noted before that in the case of factorizing transport cost (see eq. \eqref{eq:tr-cost-factorizes}), the primal task \eqref{eq:IID-primal-problem} reduces to the linear problem
\begin{align} \label{eq:primal-factorized-2}
    \text{minimize } \Pi \mapsto \tr_{\hohc} \lesq{\Pi \ler{\sum_{k=1}^K \iint_{\R \times \R} f_k\ler{x_k,y_k} \dd  E_k(y_k) \otimes \dd E_k^T (x_k)}} \text{ over } \cC\ler{\rho,\omega}.
\end{align}
Let us consider the special case $K=1$ in Theorem \ref{thm:linear-duality}, and let us replace the cost operator $C_c^{(\cA)}$ there by an arbitrary self-adjoint operator $\tilde{C}$ acting on $\hohc.$
Observe that the concrete form of the cost operator $C_c^{(\cA)}$ does not play any role in the proof of Theorem \ref{thm:linear-duality}, the cost operator can be replaced by any self-adjoint operator $\tilde{C}$. Consequently, the proof of Theorem \ref{thm:linear-duality} shows that one gets strong Kantorovich duality also for the primal problem \eqref{eq:primal-factorized-2} where $\tilde{C}=\sum_{k=1}^K \iint_{\R \times \R} f_k\ler{x_k,y_k} \dd  E_k(y_k) \otimes \dd E_k^T (x_k)$, which we formalize in the following corollary.

\begin{corollary} \label{cor:factorized-duality}
Assume that the transport cost factorizes in the sense of \eqref{eq:tr-cost-factorizes}. In this case, the primal problem \eqref{eq:IID-primal-problem} admits a strong Kantorovich dual problem, which is to maximize
$\tr_{\cH}\lesq{\omega Y}+\tr_{\cH}\lesq{\rho X}$ under the constraint $Y \otimes I_{\cH}^T +I_{\cH} \otimes X^T \leq C_{fac}:=\sum_{k=1}^K \iint_{\R \times \R} f_k\ler{x_k,y_k} \dd  E_k(y_k) \otimes \dd E_k^T (x_k).$ 
That is, 
\begin{align}
    \sup \lers{ \tr_{\cH}\lesq{\omega Y}+\tr_{\cH}\lesq{\rho X} \, \middle| \, Y \otimes I_{\cH}^T +I_{\cH} \otimes X^T \leq C_{fac}}
    =\min \lers{\tr_{\hohc} \lesq{\Pi C_{fac}} \, \middle| \, \Pi \in \cC\ler{\rho, \omega}},
\end{align}
where the variables $X$ and $Y$ to be optimized are self-adjoint and bounded operators on $\cH.$
\end{corollary}

Interestingly, optimal quantum Kantorovich potentials might not exist, even in low dimensional cases. In other words, it can happen that there is no maximizer for the dual problem. Indeed, let us consider the duality described in Theorem \ref{thm:linear-duality} in the simple case when $\cH=\C^2, \, K=1, \, \cA=\lers{\sigma_1},$ where $\sigma_1$ is one of the Pauli matrices defined in \eqref{eq:Pauli-def}, $c(t,t')=(t-t')^2, \, \rho \in \cS(\C^2)$ is an arbitrary full-rank density matrix, and $\omega=\frac{1}{2}(I+\sigma_x)=\lesq{\ba{cc} 1 & 0 \\ 0 & 0 \ea}.$ The state $\omega$ is pure, and hence the only and consequently optimal coupling of $\rho$ and $\omega$ is $\Pi:=\omega \otimes \rho^T.$ The transport cost operator is 
\begin{align}
C_{(t-t')^2}^{(\{\sigma_1\})}=\ler{\sigma_1 \otimes I^T -I \otimes \sigma_1^T}^2 
=\lesq{ \ba{cccc} 2 & 0 & 0 & -2 \\ 0 & 2 & -2 & 0 \\ 0 & -2 & 2 & 0 \\ -2 & 0 & 0 & 2\ea}. \nonumber
\end{align}
Assume indirectly that the dual problem admits a maximizer, that is, there exist $X=\lesq{\ba{cc} x_{11} & x_{12} \\ \overline{x_{12}} & x_{22} \ea} \in \cB(\C^2)^{sa}$ and $Y=\lesq{\ba{cc} y_{11} & y_{12} \\ \overline{y_{12}} & y_{22} \ea} \in \cB(\C^2)^{sa}$ such that
\begin{align}
    \tr_{\C^2}[\omega Y] +\tr_{\C^2}[\rho X]=\tr_{\C^2 \otimes \ler{\C^2}^*}\lesq{\Pi C_{(t-t')^2}^{(\{\sigma_1\})}},
    \nonumber
\end{align}
and $Y \otimes I^T +I \otimes X^T \leq C_{(t-t')^2}^{(\{\sigma_1\})}.$ Equivalently, for the operator $M:=C_{(t-t')^2}^{(\{\sigma_1\})}-\ler{Y \otimes I^T +I \otimes X^T}$ we have $\tr_{\C^2 \otimes \ler{\C^2}^*}\lesq{\Pi M}=0$ and $M \geq 0.$ Let us write $\Pi$ and $M$ in their block-matrix forms, that is, 
\begin{align}
\Pi=\omega \otimes \rho^T = \lesq{ \ba{c|c} \rho^T & 0 \\ \hline 0 & 0  \ea}
\text{ and }
M=\lesq{ \ba{c|c} M_{11} & M_{12} \\ \hline M_{12}^* & M_{22}  \ea}.
    \nonumber
\end{align}
It is immediately seen by these forms that $\tr_{\C^2 \otimes \ler{\C^2}^*}\lesq{\Pi M}=0$ is equivalent to $\tr_{\ler{\C^2}^*}\lesq{\rho^T M_{11}}=0,$ which implies that $M_{11}=0$ because both $\rho^T$ and $M_{11}$ are positive semi-definite, and $\rho^T$ is full-rank. But then the condition $M \geq 0$ implies that $M_{12}=M_{12}^*=0.$ However, this cannot happen as 
\begin{align}
    M=\lesq{ \ba{cccc} 2 & 0 & 0 & -2 \\ 0 & 2 & -2 & 0 \\ 0 & -2 & 2 & 0 \\ -2 & 0 & 0 & 2\ea} 
    - \lesq{ \ba{cccc} y_{11} & 0 & y_{12} & 0 \\ 0 & y_{11} & 0 & y_{12} \\ \overline{y_{12}} & 0 & y_{22} & 0 \\ 0 & \overline{y_{12}} & 0 & y_{22}\ea}
    - \lesq{ \ba{cccc} x_{11} & x_{12} & 0 & 0 \\ \overline x_{12} & x_{22} & 0 & 0 \\ 0 & 0 & x_{11} & x_{12} \\ 0 & 0 & \overline{x_{12}} & x_{22}\ea},
    \nonumber
\end{align}
and hence $[M]_{1,4},$ the entry in the top-right corner of $M$ is $-2$ no matter what $X$ and $Y$ are. We got a contradiction, as desired. A similar reasoning shows that there is no maximizer either if we modify the above example by setting $\rho=\omega=\frac{1}{2}(I+\sigma_1).$

In \cite[Section 2.2]{BPTV-p-Wass} we considered also the following quantum mechanical optimal transport problem: let $\cH:=L^2(\R^K) \simeq L^2(\R)^{\otimes K},$ and let $c: \R^K \times \R^K \to [0,\infty)$ be a non-negative lower semi-continuous classical cost function. Let $E: \cB(\R)\to \cP(L^2(\R))$ be the spectral measure of the position operator $Q$ acting on $L^2(\R),$ that is, $E(S)=M_{\chi_S},$ where $\chi_S$ is the characteristic function of $S$ and $M_f$ is the multiplication by $f$ given by $(M_f \psi)(x)=f(x)\psi(x).$
\par
The cost operator $C_c \in \mathrm{Lin}\ler{L^2(\R^K) \otimes (L^2(\R^K))^*}$ corresponding to the classical cost $c$ is defined by Borel functional calculus the following way:
\begin{align} \label{eq:pos-cost-op-def}
    C_c=\iint_{\R^K \times \R^K} c(x_1, \dots, x_K, y_1, \dots, y_K) \dd E(y_1) \otimes \dots \otimes \dd E(y_K) \otimes \dd E(x_1)^T \otimes \dots \otimes \dd E(x_K)^T.
\end{align}
Note that $C_c$ is unbounded if $c$ is so. Let $\rho$ and $\omega$ be states on $L^2(\R^K).$ The optimization task is to
\begin{align} \label{eq:pos-primal-task}
\text{minimize } \tr_{L^2(\R^K) \otimes (L^2(\R^K))^*}\lesq{\Pi C_c}
\end{align}
under the constraints 
\begin{align} \label{eq:pos-primal-constraints}
\Pi \in \cS\ler{L^2(\R^K) \otimes (L^2(\R^K))^*}, \, \tr_{(L^2(\R^K))^*}[\Pi]=\omega, \, \tr_{L^2(\R^K)} [\Pi]=\rho^T.
\end{align}

Just like in the case of Corollary \ref{cor:factorized-duality}, the proof of Theorem \ref{thm:linear-duality} with $K=1$ and with the appropriate cost operator demonstrates that the primal quantum optimal transport problem described in \eqref{eq:pos-primal-task} and \eqref{eq:pos-primal-constraints} has a strong dual. We formalize the precise statement in the following corollary.

\begin{corollary} \label{cor:L2RK-duality}
    Let the cost operator $C_c \in \mathrm{Lin}\ler{L^2\ler{\R^K}\otimes \ler{L^2\ler{\R^K}}^*}$ be defined as in \eqref{eq:pos-cost-op-def}, and let $\rho, \omega \in \cS\ler{L^2\ler{\R^K}}.$ Then
    \begin{align} \label{eq:duality-nr-2}
        \sup \lers{ \tr_{L^2\ler{\R^K}}[\omega Y]+\tr_{L^2\ler{\R^K}}[\rho X] \,        
        \middle| \, X,Y \in \cB(\cH), \, Y \otimes I^T +I \otimes X^T \leq C_c}= \nonumber \\
        =\min \lers{\tr_{L^2\ler{\R^K}\otimes \ler{L^2\ler{\R^K}}^*}\lesq{\Gamma C_c} \, \middle| \, \Gamma \in \cS\ler{L^2\ler{\R^K}\otimes \ler{L^2\ler{\R^K}}^*}, \, \ler{\Gamma}_1=\omega, \, \ler{\Gamma}_2=\rho^T}.
    \end{align}
\end{corollary} 
The following statement demonstrates that the minimum of the primal problem \eqref{eq:IID-primal-problem} can indeed be larger than the minimum of its linear relaxation \eqref{eq:linear-primal-problem}.
\begin{proposition} \label{prop:strictly-smaller}
    There exists $C_{c}^{(\cA)}$ defined as in \eqref{eq:C-c-A-def} and states $\rho,\omega \in \sh$, such that the infimum of the primal problem defined in \eqref{eq:linear-primal-problem} is strictly smaller than the infimum of the primal problem defined in \eqref{eq:IID-primal-problem}.
\end{proposition}
\begin{proof}
    Let $\cH=\C^2,$ and with the notations introduced at the beginning of this section, let $K=3,$ and 
    $$c(x_1,x_2,x_3,y_1,y_2,y_3):=\abs{x_1-y_1}^p+\abs{x_2-y_2}^p+\abs{x_3-y_3}^p$$ for some parameter $p \geq 1.$ Let $\cA=\{\sigma_1, \sigma_2, \sigma_3\},$ where 
    \begin{align} \label{eq:Pauli-def}
        \sigma_1=\sigma_x=\lesq{\ba{cc} 0 & 1 \\ 1 & 0 \ea}, \, \quad
        \sigma_2=\sigma_y=\lesq{\ba{cc} 0 & -i \\ i & 0 \ea}, \, \quad
        \sigma_3=\sigma_z=\lesq{\ba{cc} 1 & 0 \\ 0 & -1 \ea}, 
    \end{align}
    that is, we set $\cA$ to be the collection of the Pauli matrices. Finally, let 
    \begin{align} \label{eq:rho-omega-concrete}
    \rho:=1/2(I+1/2\sigma_z) \text{ and } \omega:=1/2(I-1/2\sigma_z).    
    \end{align}
    The cost operator $C_c^{(\cA)}$ given by \eqref{eq:C-c-A-def} factorizes now the following way:
    \begin{align} \label{eq:C-c-A-factorized}
        C_{c}^{(\cA)}&=\sum_{k=1}^3 I_{\hohc}^{\otimes (k-1)} \otimes \abs{\sigma_k\otimes I^T-I\otimes \sigma_k^T}^p \otimes I_{\hohc}^{\otimes (3-k)}.
    \end{align}
    Thus, on one hand, as we noted in \eqref{eq:primal-factorized} and \eqref{eq:primal-factorized-2}, the task \eqref{eq:IID-primal-problem} takes the form
    \begin{equation}\begin{aligned}
        \text{minimize } \tr_{\hohc} \lesq{\Pi \ler{\sum_{k=1}^3 \abs{\sigma_k\otimes I^T-I\otimes \sigma_k^T}^p}}
    \end{aligned}\end{equation}
    where $\Pi$ runs over the set of all couplings of $\rho, \omega \in \sh.$ On the other hand, as we noted in \eqref{eq:linear-factorized}, the task \eqref{eq:linear-primal-problem} takes the form
    \begin{equation}\begin{aligned}
        \text{minimize } \sum_{k=1}^3\tr_{\hohc} \lesq{\Pi_k \abs{\sigma_k\otimes I^T-I\otimes \sigma_k^T}^p }
    \end{aligned}\end{equation}
    where all $\Pi_k$ run over the set of all couplings of $\rho, \omega \in \sh.$ Taking into account the concrete form of $\rho$ and $\omega$ (see \eqref{eq:rho-omega-concrete}), we conclude that the minimum of \eqref{eq:IID-primal-problem} takes the form
    \begin{align}
        \min \lers{\tr_{\hohc}\lesq{\Pi \ler{\sum_{k=1}^3 \abs{\sigma_k\otimes I^T-I\otimes \sigma_k^T}^p}} \, \middle| \, \Pi \in \cC\ler{1/2(I+1/2\sigma_z), 1/2(I-1/2\sigma_z)}} \nonumber \\
        =2^p\left(1+\frac{1}{2}-\sqrt{\left(1-\frac{1}{2}\right)\left(1-\frac{1}{2}\right)}\right)=2^p,    
    \end{align}
    where we made use of the fact that $\rho$ and $\omega$ given by \eqref{eq:rho-omega-concrete} commute, and used Theorem~\ref{thm:D_symm_for_commuting_states} from the subsequent section where we give an explicit closed form for the optimal transport cost between commuting states. On the other hand, using again the the concrete form of $\rho$ and $\omega$ we conclude that the minimum of \eqref{eq:linear-primal-problem} takes the following form
     \begin{align}\label{eq:genprimal_halfway}
       \min \lers{\sum_{k=1}^3\tr_{\hohc}\lesq{\Pi_k \abs{\sigma_k\otimes I^T-I\otimes \sigma_k^T}^p} \, \middle| \, \Pi_1,\Pi_2, \Pi_3 \in \cC\ler{1/2(I+1/2\sigma_z), 1/2(I-1/2\sigma_z)}} \nonumber \\
       =\sum_{k=1}^3
       \min\lers{ \tr_{\hohc} \lesq{\Pi_k \abs{\sigma_k\otimes I^T-I\otimes \sigma_k^T}^p} \middle| \, \Pi_k \in \cC\ler{1/2(I+1/2\sigma_z), 1/2(I-1/2\sigma_z)}}.
    \end{align}
    The first two terms of the sum on the right-hand side of \eqref{eq:genprimal_halfway} can be computed explicitly by Theorem~\ref{thm:D_z_for_xy_states} of the next Section (with appropriate changes of basis), while the third term is given by Proposition~\ref{prop:D_z_for_z_states} there. Accordingly,
    \begin{align} 
    \sum_{k=1}^3
       \min\lers{ \tr_{\hohc} \lesq{\Pi_k \abs{\sigma_k\otimes I^T-I\otimes \sigma_k^T}^p} \middle| \, \Pi_k \in \cC\ler{1/2(I+1/2\sigma_z), 1/2(I-1/2\sigma_z)}} \nonumber \\
       =2^p\left(1-\sqrt{1-\frac{1}{2^2}}\right)+2^{p-1}=2^p-\left(\frac{\sqrt{3}-1}{2}\right)2^p<2^p,
    \end{align}
    which completes the proof.
\end{proof}

\section{Examples of strong duality achieved and applications}

In this section, we will apply the Kantorovich duality results Theorem \ref{thm:linear-duality} and Corollary \ref{cor:factorized-duality} to prove the optimality of certain quantum couplings and operator Kantorovich potentials. We consider the case of quantum bits, that is, $\cH=\C^2,$ and the following transportation costs will be studied (with the notation introduced at the beginning of Section \ref{sec:Kantorovich-duality}):
\begin{enumerate}
    \item $K=3, \, \cA=\lers{\sigma_1, \sigma_2,\sigma_3},$ and $c(x,y)=\norm{x-y}_p^p,$ where $x,y \in \R^3,$ and $\norm{\cdot}_p$ is the $l_p$ norm there;
    \item $K=1, \, \cA=\lers{\sigma_3},$ and $c(x,y)=\abs{x-y}^p,$ where $x,y \in \R.$
\end{enumerate}

\subsection{Strong duality for commuting qubits and symmetric transport cost}
Let $K=3, \, \cA=\lers{\sigma_1, \sigma_2,\sigma_3},$ and $c(x,y)=\norm{x-y}_p^p$ for some parameter $p\geq 1.$ According to \eqref{eq:C-c-A-def}, the cost operator $C_c^{(\cA)}$ is the one given in \eqref{eq:C-c-A-factorized}, and the primal quantum optimal transport problem \eqref{eq:IID-primal-problem} reduces to
\begin{align}
    \text{minimize } \tr_{\hohc}\lesq{\Pi C_{\text{symm},p}} \text{ over all } \Pi \in \cC\ler{\rho, \omega}, 
\end{align}
where 
\begin{align}
    C_{\text{symm},p}=\sum_{k=1}^3 \abs{\sigma_k\otimes I^T-I\otimes \sigma_k^T}^p.
\end{align}
The cost operator $C_{\text{symm},p}$ can be computed explicitly:
    \begin{align}\label{eq:Csymm}
        C_{\text{symm},p}=2^{p+1}I\otimes I^T -2^p\Ket{I}\Bra{I}
        =\lesq{\ba{cccc}
        2^{p} & 0 & 0 & -2^{p}\\
        0 & 2^{p+1} & 0 & 0\\
        0 & 0 & 2^{p+1} & 0\\
        -2^{p} & 0 & 0 & 2^{p}\ea}.
    \end{align}
The matrix form of the symmetric cost operator $C_{\text{symm},p}$ is in fact basis-invariant, that is, 
\begin{align} \label{eq:C-symm-unitary-invariance}
    \ler{U \otimes \ler{U^*}^T} C_{\text{symm},p} \ler{U^* \otimes U^T}=C_{\text{symm},p}
\end{align}
for every unitary $U$ acting on $\C^2.$ Consequently, for commuting quantum bits one can assume without loss of generality that both qubits commute with $\sigma_z$.
\par
In the following Proposition \ref{prop:D_symm_for_commuting_states} and Theorem \ref{thm:D_symm_for_commuting_states} we determine the optimal couplings of commuting quantum bits with respect to the transportation cost described by $C_{\text{symm},p},$ and we give a simple closed form for the induced $p$-Wasserstein distance $D_{\text{symm},p}.$ We recall that according to the recipe given in \cite[Section 3]{BPTV-p-Wass}, the $p$-Wasserstein distance $D_{\text{symm},p}$ corresponding to the cost operator $C_{\text{symm},p}$ is defined by
\begin{align} \label{eq:D-symm-p-def}
    D_{\text{symm},p}=\ler{\min_{\Pi \in \cC(\rho, \omega)} \lers{\tr_{\hohc}\lesq{\Pi C_{\text{symm},p}}}}^{\frac{1}{p}}.
\end{align}

\begin{proposition}\label{prop:D_symm_for_commuting_states}
    Let
    \begin{equation}\label{eq:rho_alpha}\begin{aligned}
        \rho(\alpha):=\frac{1}{2}\left(I+\alpha\sigma_z\right)=\begin{pmatrix}\frac{1+\alpha}{2} & 0 \\ 0 & \frac{1 - \alpha}{2} \end{pmatrix},
    \end{aligned}\end{equation}
    for $\alpha \in [-1,1].$ Then the optimal coupling of $\rho(\alpha)$ and $\rho(\beta)$ is given by \eqref{eq:rho-alpha-beta-optimal-coupling}, and    
    \begin{equation}\label{eq:D_comm}\begin{aligned}
        D^p_{\text{symm},p}(\rho(\alpha),\rho(\beta))=2^p\left(1+\frac{1}{2}\abs{\alpha-\beta}-\sqrt{(1+\min(\alpha,\beta))(1-\max(\alpha,\beta))}\right).
    \end{aligned}\end{equation}
\end{proposition}

\begin{proof}
    By using the symmetry mentioned above, one could also assume without losing generality that e.g. $\alpha\geq\beta$ and then arrive to \eqref{eq:D_comm} without the extrema. Instead, for completeness we will prove \eqref{eq:D_comm} directly. Let $z_-:=\min(\alpha,\beta)$ and $z_+:=\max(\alpha,\beta)$, then let
    \begin{align} \label{eq:rho-alpha-beta-optimal-coupling}
        \Pi(\alpha,\beta):=\frac{1}{2}\lesq{\ba{cccc}
        1+z_{-} & 0 & 0 & \sqrt{(1+z_-)(1-z_+)}\\
        0 & \max(\beta-\alpha,0) & 0 & 0\\
        0 & 0 & \max(\alpha-\beta,0) & 0\\
        \sqrt{(1+z_-)(1-z_+)} & 0 & 0 & 1-z_+\ea}.
    \end{align}
    The matrix $\Pi(\alpha,\beta)$  is clearly hermitian and is positive-semidefinite by Sylvester's criterion, since all principal minors of $\Pi(\alpha,\beta)$ are nonnegative. It is easy to check that $\tr_1 \lesq{\Pi(\alpha,\beta)}=\rho(\alpha)^T$ while $\tr_2 \lesq{ \Pi(\alpha,\beta)}=\rho(\beta)$ (and consequently, $\tr \lesq{\Pi(\alpha,\beta)}=1$), which demonstrate that $\Pi(\alpha,\beta)$ is a coupling of $\rho(\alpha)$ and $\rho(\beta).$
    It follows that
    \begin{align} \label{eq:D_upperbound}
         D^p_{\text{symm},p}(\rho(\alpha),\rho(\beta))&\leq \tr \lesq{C_{\text{symm},p}\Pi(\alpha,\beta)}
         =2^{p+1}-2^p\bra{\bra{I}}\Pi(\alpha,\beta)\ket{\ket{I}} \nonumber \\
         &=2^{p+1}-2^{p-1}\left((1+z_-)+(1-z_+)+2\sqrt{(1+z_-)(1-z_+)} \right)\\
         &=2^p\left(1+\frac{1}{2}\abs{\alpha-\beta}-\sqrt{(1+\min(\alpha,\beta))(1-\max(\alpha,\beta))}\right).
    \end{align}
    On the other hand, if $\abs{\alpha}\neq 1$ and $\abs{\beta}\neq 1$ consider
    \begin{equation}\begin{aligned}
        X_1=\lesq{\ba{cc} -2^p\sqrt{\frac{1-\beta}{1+\alpha}}-2^p & 0 \\ 0 & 0 \ea},\quad Y_1=\lesq{\ba{cc} 2^{p+1} & 0 \\ 0 & 2^p-2^p\sqrt{\frac{1+\alpha}{1-\beta}} \ea},
    \end{aligned}\end{equation}
    and
    \begin{equation}\begin{aligned}
        X_2=\lesq{\ba{cc} 2^{p+1} & 0 \\ 0 & 2^p-2^p\sqrt{\frac{1+\beta}{1-\alpha}} \ea},\quad Y_2=\lesq{\ba{cc} -2^p\sqrt{\frac{1-\alpha}{1+\beta}}-2^p & 0 \\ 0 & 0 \ea}.
    \end{aligned}\end{equation}
    Clearly, $X_1$, $X_2$, $Y_1$ and $Y_2$ are self-adjoint. It is also evident by Sylvester's criterion, that
    \begin{equation}\begin{aligned}
        C_{\text{symm},p}-Y_1\otimes I^T-I\otimes X_1^T=
        \lesq{\ba{cccc}2^p\sqrt{\frac{1-\beta}{1+\alpha}} & 0 & 0 & -2^p\\
        0 & 0 & 0 & 0\\
        0 & 0 & 2^{p+1}+2^p\sqrt{\frac{1-\beta}{1+\alpha}}+2^p\sqrt{\frac{1+\alpha}{1-\beta}} & 0\\
        -2^p & 0 & 0 & 2^p\sqrt{\frac{1+\alpha}{1-\beta}}\ea}\geq 0,
    \end{aligned}\end{equation}
    and
    \begin{equation}\begin{aligned}
        C_{\text{symm},p}-Y_2\otimes I^T-I\otimes X_2^T=
        \lesq{\ba{cccc}2^p\sqrt{\frac{1-\alpha}{1+\beta}} & 0 & 0 & -2^p\\
        0 & 2^{p+1}+2^p\sqrt{\frac{1-\alpha}{1+\beta}}+2^p\sqrt{\frac{1+\beta}{1-\alpha}} & 0 & 0\\
        0 & 0 & 0 & 0\\
        -2^p & 0 & 0 & 2^p\sqrt{\frac{1+\beta}{1-\alpha}}\ea}\geq 0.
    \end{aligned}\end{equation}
    Therefore,
    \begin{equation}\begin{aligned}\label{eq:D_lowerbound}
        &D^p_{\text{symm},p}(\rho(\alpha),\rho(\beta))\geq  \max\lers{\tr\lesq{ X_1\rho(\alpha)}+\tr \lesq{ Y_1\rho(\beta)},\tr \lesq{X_2\rho(\alpha)}+\tr \lesq{Y_2\rho(\beta)}}\\
        =&\max\lers{-2^{p}\sqrt{(1+\alpha)(1-\beta)}+2^p+2^{p-1}(\beta-\alpha),-2^{p}\sqrt{(1+\beta)(1-\alpha)}+2^p+2^{p-1}(\alpha-\beta)}\\
        =&2^p\left(1+\frac{1}{2}\abs{\alpha-\beta}-\sqrt{(1+\min(\alpha,\beta))(1-\max(\alpha,\beta))}\right).
    \end{aligned}\end{equation}
    For mixed states, combining \eqref{eq:D_upperbound} and \eqref{eq:D_lowerbound} completes the proof. If either state is pure then it is known that there is only one coupling, the tensor product, and therefore \eqref{eq:D_upperbound} is an equality rather than an upper bound. 
\end{proof}

One might observe the clear structure of the optimal coupling $\Pi(\alpha, \beta)$ in \eqref{eq:rho-alpha-beta-optimal-coupling}. If we consider the two-point space $S=\{s, s'\}$ with an arbitrary bona fide metric and associate the classical probability distribution $\mu(\alpha)=\frac{1+\alpha}{2} \delta_{q} + \frac{1-\alpha}{2} \delta_{q'}$ to $\rho(\alpha)$ (for all $\alpha \in [-1,1]$), and solve the classical optimal transport problem between the measures $\rho(\alpha)$ and $\rho(\beta),$ then we get that the optimal transport plan is $\pi_0=\frac{1+z_{-}}{2} \delta_{s,s}+\frac{\max(\alpha-\beta,0)}{2} \delta_{s,s'}+\frac{\max(\beta-\alpha,0)}{2} \delta_{s',s}+\frac{1-z_{+}}{2} \delta_{s',s'}.$ Copying this result in the diagonal of a $4\times 4$ matrix in an appropriate way and then choosing the modulus of the top-right entry as large a possible gives the optimal coupling $\Pi(\alpha, \beta).$
\par
As for the shape of the optimal Kantorovich potentials, we cannot provide such a clear picture, but one might note that all the potentials are diagonal matrices, and one of the diagonal elements of $X$ or $Y$ can be chosen to be $0$ by a shift invariance property. The remaining three elements are adjusted so that $C_{\text{symm},p}-Y\otimes I^T - I \otimes X^T$ remains positive but becomes orthogonal to $\Pi(\alpha, \beta).$

\begin{theorem}\label{thm:D_symm_for_commuting_states}
    Let $\rho$ denote now the standard Bloch parametrization of quantum bits, that is,
    \begin{equation}\label{eq:rho_r}\begin{aligned}
        \rho(\vec{r}):=\frac{1}{2}\left(I+\vec{r}\cdot \vec{\sigma}\right), \text{ where } \vec{\sigma}=\ler{\sigma_1,\sigma_2, \sigma_3},
    \end{aligned}\end{equation}
    and let us assume that $\vec{r}_1$ and $\vec{r}_2$ are scalar multiples of each other implying that $\rho(\vec{r}_1)$ and $\rho(\vec{r}_2)$ commute. Then
    \begin{equation}\label{eq:D_comm_gen}\begin{aligned}
        D^p_{\text{symm},p}(\rho(\vec{r}_1),\rho(\vec{r}_2))
        =2^p\left(1+\frac{1}{2}\abs{\vec{r}_1-\vec{r}_2}-\sqrt{\ler{1+\frac{\vec{r}_1 \cdot \vec{r}_2}{\max\{\abs{\vec{r_1}},\abs{\vec{r_2}}\}}}\ler{1-\max\{\abs{\vec{r_1}},\abs{\vec{r_2}}\}}}\right).
    \end{aligned}\end{equation}
\end{theorem}
\begin{proof}
    Immediate from the basis-independence of \eqref{eq:Csymm} and Proposition~\ref{prop:D_symm_for_commuting_states}.
\end{proof}

It is an interesting phenomenon that, according to many of the approaches including the one we follow in the present work \cite{DPT-AHP,DPT-lecture-notes}, the quantum Wasserstein distance of states is not a bona fide metric, for example, states may have positive distance from themselves. As a response to this phenomenon, De Palma and Trevisan introduced quadratic quantum Wasserstein divergences \cite{DPT-lecture-notes}, which are appropriately modified versions of quadratic quantum Wasserstein distances, to eliminate self-distances.
Their definition of the quadratic quantum Wasserstein divergence $d_{\cA,2}$ corresponding to the collection $\cA=\lers{A_1, \dots, A_K}$ of observables is the following:
\begin{align} \label{eq:quadratic-divergence-def}
d_{\cA,2}\ler{\rho, \omega}:=\ler{D_{\cA,2}^2\ler{\rho, \omega}-\frac{1}{2}\ler{D_{\cA,2}^2\ler{\rho, \rho}+D_{\cA,2}^2\ler{\omega, \omega}}}^{\frac{1}{2}},
\end{align}
where 
\begin{align} \label{eq:quadratic-distance-def}
    D_{\cA,2}^2\ler{\rho, \omega}= \min\lers{\tr_{\hohc}\lesq{\Pi \ler{\sum_{k=1}^K \ler{A_k \otimes I^T + I \otimes A_k^T}^2}} \, \middle| \Pi \in \cC(\rho, \omega)}.
\end{align}
They conjectured that the divergences defined this way are genuine metrics on quantum state spaces \cite{DPT-lecture-notes}, and this conjecture has recently been justified under certain additional assumptions \cite{BPTV-metric-24}.

In the following corollary, we use Theorem \ref{thm:D_symm_for_commuting_states} to obtain a closed form for the quadratic divergence $d_{\text{symmm},2}=d_{\lers{\sigma_1,\sigma_2,\sigma_3},2}.$

\begin{corollary}\label{cor:d_for_commuting_states}
    Let $\rho$ denote the Bloch parametrization as in \eqref{eq:rho_r}, let the $2$-Wasserstein distance $D_{\text{symm},2}$ be given by \eqref{eq:D-symm-p-def}, and let the corresponding quadratic Wasserstein divergence $d_{\text{symm},2}$ be given by \eqref{eq:quadratic-divergence-def}.
    Assume that $\vec{r}_2$ is a scalar multiple of $\vec{r}_1$ and hence $\rho(\vec{r}_1)$ and $\rho(\vec{r}_2)$ commute. Then
    \begin{align}
        d^2_{\text{symm},2}(\rho(\vec{r}_1),\rho(\vec{r}_2))  \nonumber \\
        =2\left(\abs{\vec{r}_1-\vec{r}_2}+\sqrt{1-r_1^2}+\sqrt{1-r_2^2}-2\sqrt{\ler{1+\frac{\vec{r}_1 \cdot \vec{r}_2}{\max\{\abs{\vec{r_1}},\abs{\vec{r_2}}\}}}\ler{1-\max\{\abs{\vec{r_1}},\abs{\vec{r_2}}\}}}\right).
    \end{align}
\end{corollary}
\begin{proof}
Direct computation shows that
    \begin{equation}\begin{aligned}
         \empty &d^2_{\text{symm},2}(\rho(\vec{r}_1),\rho(\vec{r}_2))\\
         =&D^2_{\text{symm},2}(\rho(\vec{r}_1),\rho(\vec{r}_2))-\frac{1}{2}\Tr C_{\text{symm},2}\left(\ket{\ket{\sqrt{\rho(\vec{r}_1)}}}\bra{\bra{\sqrt{\rho(\vec{r}_1)}}}+\ket{\ket{\sqrt{\rho(\vec{r}_2)}}}\bra{\bra{\sqrt{\rho(\vec{r}_2)}}}\right)\\
        =&D^2_{\text{symm},2}(\rho(\vec{r}_1),\rho(\vec{r}_2))-\frac{1}{2} \left(\bra{\bra{\sqrt{\rho(\vec{r}_1)}}}C_{\text{symm},2}\ket{\ket{\sqrt{\rho(\vec{r}_1)}}}+\bra{\bra{\sqrt{\rho(\vec{r}_2)}}}C_{\text{symm},2}\ket{\ket{\sqrt{\rho(\vec{r}_2)}}}\right)\\
        =&D^2_{\text{symm},2}(\rho(\vec{r}_1),\rho(\vec{r}_2))-\frac{1}{2} \left(2^3-2^2\abs{\left<\left<I\middle\|\sqrt{\rho(\vec{r}_1)}\right>\right>}+2^3-2^2\abs{\left<\left<I\middle\|\sqrt{\rho(\vec{r}_2)}\right>\right>}\right)\\
        =&D^2_{\text{symm},2}(\rho(\vec{r}_1),\rho(\vec{r}_2))-2^3+2 \left(\left[\Tr\sqrt{\rho(\vec{r}_1)}\right]^2+\left[\Tr\sqrt{\rho(\vec{r}_2)}\right]^2\right)\\
        =&D^2_{\text{symm},2}(\rho(\vec{r}_1),\rho(\vec{r}_2))-2^3+2 \left(\left[\sqrt{\frac{1+r_1}{2}}+\sqrt{\frac{1-r_1}{2}}\right]^2+\left[\sqrt{\frac{1+r_2}{2}}+\sqrt{\frac{1-r_2}{2}}\right]^2\right)\\
        =&D^2_{\text{symm},2}(\rho(\vec{r}_1),\rho(\vec{r}_2))-2^2+2\sqrt{1-r_1^2}+2\sqrt{1-r_2^2}\\
        =&2^2\left(1+\frac{1}{2}\abs{\vec{r}_1-\vec{r}_2}-\sqrt{(1+\frac{\vec{r}_1\vec{r}_2}{\max(r_1,r_2)})(1-\max(r_1,r_2))}\right)-2^2+2\sqrt{1-r_1^2}+2\sqrt{1-r_2^2}\\
        =&2\left(\abs{\vec{r}_1-\vec{r}_2}+\sqrt{1-r_1^2}+\sqrt{1-r_2^2}-2\sqrt{(1+\frac{\vec{r}_1\vec{r}_2}{\max(r_1,r_2)})(1-\max(r_1,r_2))}\right),
    \end{aligned}\end{equation}
    where we used Theorem \ref{thm:D_symm_for_commuting_states} in the penultimate equality, and we also used that the distance of a state from itself is given by $D_{\text{symm},2}^2(\rho,\rho)=\tr\lesq{C_{\text{symm},2} \ket{\ket{\sqrt{\rho}}}\bra{\bra{\sqrt{\rho}}}},$ see \cite[Corollary 1]{DPT-AHP}.
\end{proof}

Using the above obtained closed formula for the quadratic Wasserstein divergence $d_{\text{symm},2},$ we prove in the next proposition that even the squared quantity $d_{\text{symm},2}^2$ satisfies the triangle inequality if all three qubits involved commute with each other.

\begin{proposition}
    For commuting qubits $\rho,\sigma,\omega \in \cS\ler{\C^2}$ the triangle inequality 
    \begin{equation}\begin{aligned}
        d^2_{\text{symm},2}(\rho,\sigma)+d^2_{\text{symm},2}(\sigma,\omega) \geq d^2_{\text{symm},2}(\rho,\omega)
    \end{aligned}\end{equation}
    holds.
\end{proposition}
\begin{proof}
    For commuting $\rho,\sigma,\omega$ it can be assumed that there are real numbers $-1\leq \alpha,\beta,\gamma\leq 1$, for which $\rho=\rho(\alpha),\sigma=\rho(\beta),\omega=\rho(\gamma)$ as in \eqref{eq:rho_alpha}. Thus by Corollary~\ref{cor:d_for_commuting_states} we have that
    \begin{equation}\begin{aligned}\label{eq:triang_ineq}
        &\frac{1}{2}\left(d^2_{\text{symm},2}(\rho,\sigma)+d^2_{\text{symm},2}(\sigma,\omega)-d^2_{\text{symm},2}(\rho,\omega)\right)\\=&\abs{\alpha-\beta}+\sqrt{1-\alpha^2}+\sqrt{1-\beta^2}-2\sqrt{(1+\min(\alpha,\beta))(1-\max(\alpha,\beta))}\\+&\abs{\beta-\gamma}+\sqrt{1-\beta^2}+\sqrt{1-\gamma^2}-2\sqrt{(1+\min(\beta,\gamma))(1-\max(\beta,\gamma))}\\-&\abs{\alpha-\gamma}-\sqrt{1-\alpha^2}-\sqrt{1-\gamma^2}+2\sqrt{(1+\min(\alpha,\gamma))(1-\max(\alpha,\gamma))}\\=&\abs{\alpha-\beta}+\abs{\beta-\gamma}-\abs{\alpha-\gamma}\\
        +&2\sqrt{1-\beta^2}-2\sqrt{(1+\min(\alpha,\beta))(1-\max(\alpha,\beta))}\\
        -&2\sqrt{(1+\min(\beta,\gamma))(1-\max(\beta,\gamma))}+2\sqrt{(1+\min(\alpha,\gamma))(1-\max(\alpha,\gamma))}\\
        \geq&2\sqrt{1-\beta^2}-2\sqrt{(1+\min(\alpha,\beta))(1-\max(\alpha,\beta))}\\
        -&2\sqrt{(1+\min(\beta,\gamma))(1-\max(\beta,\gamma))}+2\sqrt{(1+\min(\alpha,\gamma))(1-\max(\alpha,\gamma))}\\
    \end{aligned}\end{equation}
    where we used the triangle inequality for $d(a,b):=\abs{a-b}$.

    If $\alpha\leq\beta\leq\gamma$, then the last line of \eqref{eq:triang_ineq} takes the following form:
    \begin{equation}\begin{aligned}\label{eq:order_abc}
        2\sqrt{(1+\beta)(1-\beta)}-2\sqrt{(1+\alpha)(1-\beta)}-2\sqrt{(1+\beta)(1-\gamma)}+2\sqrt{(1+\alpha)(1-\gamma)}.
    \end{aligned}\end{equation}
    If $\alpha\leq\beta\leq\gamma$, then
    \begin{equation}\begin{aligned}
        (\gamma-\beta)(\beta-\alpha)=\gamma\beta-\beta^2-\alpha\gamma+\alpha\beta&\geq 0\quad \Leftrightarrow\\
        -\beta^2-\alpha\gamma&\geq -\gamma\beta-\alpha\beta\quad \Leftrightarrow\\
        \left(\sqrt{(1+\beta)(1-\beta)}+\sqrt{(1+\alpha)(1-\gamma)}\right)^2&\geq \left(\sqrt{(1+\alpha)(1-\beta)}+\sqrt{(1+\beta)(1-\gamma)}\right)^2\quad \Leftrightarrow\\
        \sqrt{(1+\beta)(1-\beta)}+\sqrt{(1+\alpha)(1-\gamma)}&\geq \sqrt{(1+\alpha)(1-\beta)}+\sqrt{(1+\beta)(1-\gamma)},
    \end{aligned}\end{equation}
    from which it follows that if $\alpha\leq\beta\leq\gamma$ \eqref{eq:order_abc} is nonnegative and then so is the last line of \eqref{eq:triang_ineq}.
     If $\beta\leq\alpha\leq\gamma$, then half of the last line of \eqref{eq:triang_ineq} takes the following form:
    \begin{equation}\begin{aligned}
        &\sqrt{(1+\beta)(1-\beta)}-\sqrt{(1+\beta)(1-\alpha)}-\sqrt{(1+\beta)(1-\gamma)}+\sqrt{(1+\alpha)(1-\gamma)}\\
        =&\sqrt{(1+\beta)}\left(\sqrt{1-\beta}-\sqrt{(1-\alpha)}\right)+\sqrt{1-\gamma}\left(\sqrt{(1+\alpha)}-\sqrt{(1+\beta)}\right),
    \end{aligned}\end{equation}
    which is then nonnegative by assumption.
    The other four cases of the order of $\alpha,\beta,\gamma$ can be transformed into either one of the above two with the use of variable changes $\alpha':=(1-\alpha)$, $\beta':=(1-\beta)$, $\gamma':=(1-\gamma)$ and using the fact that \eqref{eq:triang_ineq} is symmetric in $\alpha$ and $\gamma$. Thus the last line of \eqref{eq:triang_ineq} is nonnegative which completes the proof.
\end{proof}
\subsection{Strong duality for special cases of qubits and single observable cost}
In this subsection we consider the case when a single observable generates the transport cost. On quantum bits, we may assume (up to an affine rescaling of the observable and a conjugation by a unitary) that this observable is $\sigma_3=\sigma_z.$ So, we concern the setting described at the beginning of Section \ref{sec:Kantorovich-duality} and take $K=1, \, \cA=\{\sigma_z\},$ and $c(x,y)=\abs{x-y}^p.$ This choice gives rise to the cost operator 
    \begin{align}\label{eq:Cz}
        C_{z,p}:=
        C_c^{(\cA)}=
        \abs{\sigma_z\otimes I^T-I\otimes \sigma_z^T}^p=2^{p-1}\left(I\otimes I^T -\sigma_z\otimes\sigma_z^T\right)
        =\lesq{\ba{cccc} 0 & 0 & 0 & 0\\
        0 & 2^{p} & 0 & 0\\
        0 & 0 & 2^{p} & 0\\
        0 & 0 & 0 & 0\ea}
    \end{align}
where $p\geq 1.$ The transport cost $C_{z,p}$ is invariant under unitary conjugations of the form
\begin{equation}\label{eq:unit-rot-1}\begin{aligned}
    X \mapsto \ler{I\otimes \ler{\exp\ler{i\frac{\varphi}{2}\sigma_z}}^T} X \ler{I\otimes \ler{\exp\ler{-i\frac{\varphi}{2}\sigma_z}}^T}       
\end{aligned}\end{equation}
and
\begin{equation}\label{eq:unit-rot-2}\begin{aligned}
     X \mapsto \ler{\exp\ler{i\frac{\varphi}{2}\sigma_z}\otimes I^T } X \ler{\exp\ler{-i\frac{\varphi}{2}\sigma_z}\otimes I^T }   
\end{aligned}\end{equation}
however,
\begin{align}
    \Pi \in \cC(\rho,\omega)\Longleftrightarrow
    \ler{\exp\ler{i\frac{\varphi}{2}\sigma_z}\otimes I^T } \Pi \ler{\exp\ler{-i\frac{\varphi}{2}\sigma_z}\otimes I^T } \in \cC\ler{\rho, \exp\ler{i\frac{\varphi}{2}\sigma_z} \omega \exp\ler{-i\frac{\varphi}{2}\sigma_z}} \nonumber \\
    \Longleftrightarrow 
    \ler{I\otimes \ler{\exp\ler{i\frac{\varphi}{2}\sigma_z}}^T} X \ler{I\otimes \ler{\exp\ler{-i\frac{\varphi}{2}\sigma_z}}^T}
    \in\cC\ler{\exp\ler{i\frac{\varphi}{2}\sigma_z} \rho \exp\ler{-i\frac{\varphi}{2}\sigma_z},\omega}.       
\end{align}

This shows that in general whenever evaluating the $p$-Wasserstein distance \begin{align} \label{eq:D-z-p-def}
    D_{z,p}(\rho, \omega):=\ler{\min\lers{\tr_{\hohc}\lesq{\Pi C_{z,p}} \, \middle| \, \Pi \in \cC(\rho, \omega)}}^{\frac{1}{p}}
\end{align} between two qubits, one can rotate them such that neither qubit has a $\sigma_y$ coordinate anymore and compute $D_{z,p}$ then.
\par
In the following Proposition \ref{prop:D_z_for_xy_states} and Theorem \ref{thm:D_z_for_xy_states} we give a simple closed formula for the $p$-Wasserstein distance $D_{z,p}$ in the case when both qubits are orthogonal to $\sigma_z$ in the Hilbert-Schmidt sense. We obtain the formula for $D_{z,p}$ by determining the optimal transport plans and Kantorovich potentials, and we use the Kantorovich duality obtained in Section \ref{sec:Kantorovich-duality} to prove the optimality of these couplings and potentials.

\begin{proposition}\label{prop:D_z_for_xy_states}
    Let $\rho$ denote now the following reduced Bloch parametrization:
    \begin{equation}\label{eq:rho_alpha_x}\begin{aligned}
        \rho(\alpha):=\frac{1}{2}\left(I+\alpha\sigma_x\right)=\frac{1}{2}\lesq{\ba{cc}1 & \alpha \\ \alpha & 1 \ea}.
    \end{aligned}\end{equation}
    Then
    \begin{equation}\label{eq:D_z_xy}\begin{aligned}
        D^p_{z,p}(\rho(\alpha),\rho(\beta))=2^{p-1}\left(1-\sqrt{1-\max\left(\alpha^2,\beta^2\right)}\right).
    \end{aligned}\end{equation}
\end{proposition}
\begin{proof}
    By using the symmetry mentioned above, one can assume without losing generality that $\abs{\alpha}\geq\abs{\beta}$. However, for completeness we will prove \eqref{eq:D_z_xy} directly for $\abs{\alpha}<\abs{\beta}$ as well. Suppose now that $\abs{\alpha}\geq\abs{\beta}$ and $\abs{\alpha}>0$ and consider
    \begin{equation}\begin{aligned}
        \Pi_+(\alpha,\beta):=\frac{1}{4}\lesq{\ba{cccc} 1+\sqrt{1-\alpha^2} & \alpha & \beta & \frac{\left(1+\sqrt{1-\alpha^2}\right)\beta}{\alpha}\\
        \alpha & 1-\sqrt{1-\alpha^2} & \frac{\left(1-\sqrt{1-\alpha^2}\right)\beta}{\alpha} & \beta\\
        \beta & \frac{\left(1-\sqrt{1-\alpha^2}\right)\beta}{\alpha} & 1-\sqrt{1-\alpha^2} & \alpha\\
        \frac{\left(1+\sqrt{1-\alpha^2}\right)\beta}{\alpha} & \beta & \alpha & 1+\sqrt{1-\alpha^2}\ea}.
    \end{aligned}\end{equation}
    $\Pi_+(\alpha,\beta)$  is clearly hermitian and is positive-semidefinite by Sylvester's criterion. To see that all principal minors are nonnegative note that the first two columns, the last two columns, the first two rows and the last two rows are all proportional pairs, with rate $\frac{\alpha}{1+\sqrt{1-\alpha^2}}=\frac{1-\sqrt{1-\alpha^2}}{\alpha}$. It follows that the determinant and all minors of size 3 are 0-valued. All the elements in the diagonal are clearly nonnegative. Two of the principal minors of size 2 are 0-valued from the linear dependence. The nontrivial principal minors of size 2 are given by rows and columns $\{(1,3),(1,4),(2,3),(2,4)\}$. Nonnegativity for principal minors $\{(1,3),(2,4)\}$ yields the same condition
    \begin{equation}\begin{aligned}
        \left(1+\sqrt{1-\alpha^2}\right)\left(1-\sqrt{1-\alpha^2}\right)=\alpha^2\geq\beta^2,    
    \end{aligned}\end{equation}
    which is fulfilled by assumption. Nonnegativity for principal minors $\{(1,4),(2,3)\}$ yields
    \begin{equation}\begin{aligned}
        \left(1+\sqrt{1-\alpha^2}\right)^2\geq\frac{\left(1+\sqrt{1-\alpha^2}\right)^2\beta^2}{\alpha^2}\Leftrightarrow1\geq\frac{\beta^2}{\alpha^2}\Leftrightarrow\left(1-\sqrt{1-\alpha^2}\right)^2\geq\frac{\left(1-\sqrt{1-\alpha^2}\right)^2\beta^2}{\alpha^2},    
    \end{aligned}\end{equation}
    which is then again fulfilled by assumption. Easy computations show that $\Tr_1 \lesq{\Pi_+(\alpha,\beta)}=\rho(\alpha)^T$, while $\Tr_2 \lesq{\Pi_+(\alpha,\beta)}=\rho(\beta),$ which means that $\Pi_+(\alpha,\beta)$ is a coupling of $\rho(\alpha)$ and $\rho(\beta)$.
    It follows that whenever $\abs{\alpha}\geq\abs{\beta}$ and $\abs{\alpha}>0$,
    \begin{equation}\begin{aligned} \label{eq:D_z_upperbound+}
         D^p_{z,p}(\rho(\alpha),\rho(\beta))&\leq \Tr C_{z,p}\Pi_+(\alpha,\beta)=2^{p-1}\left(1-\sqrt{1-\alpha^2}\right).
    \end{aligned}\end{equation}
    If $\abs{\alpha}<\abs{\beta}$, then let us define
    \begin{equation}\begin{aligned}
        \Pi_-(\alpha,\beta):=(\Pi_+(\beta,\alpha))^{ST}=\frac{1}{4}\lesq{\ba{cccc}1+\sqrt{1-\beta^2} & \alpha & \beta & \frac{\left(1+\sqrt{1-\beta^2}\right)\alpha}{\beta}\\
        \alpha & 1-\sqrt{1-\beta^2} & \frac{\left(1-\sqrt{1-\beta^2}\right)\alpha}{\beta} & \beta\\
        \beta & \frac{\left(1-\sqrt{1-\beta^2}\right)\alpha}{\beta} & 1-\sqrt{1-\beta^2} & \alpha\\
        \frac{\left(1+\sqrt{1-\beta^2}\right)\alpha}{\beta} & \beta & \alpha & 1+\sqrt{1-\beta^2}\ea},
    \end{aligned}\end{equation}
    where  $(\cdot)^{ST}$ denotes the swap transposition on $\cT_1\ler{\hohc},$ which is the linear extension of the map $A \otimes B^T \mapsto B \otimes A^T.$
    The state $\Pi_-(\alpha,\beta)$ is a coupling of $\rho(\beta)$ and $\rho(\alpha)$. It follows that whenever $\abs{\alpha}<\abs{\beta}$,
    \begin{equation}\begin{aligned} \label{eq:D_z_upperbound-}
         D^p_{z,p}(\rho(\alpha),\rho(\beta))&\leq \Tr C_{z,p}\Pi_-(\alpha,\beta)=2^{p-1}\left(1-\sqrt{1-\beta^2}\right).
    \end{aligned}\end{equation}
    If $\alpha=\beta=0$, then
    \begin{equation}\begin{aligned}
        \Pi_0:=\lesq{\ba{cccc}\frac{1}{2} & 0 & 0 & 0\\
        0 & 0 & 0 & 0\\
        0 & 0 & 0 & 0\\
        0 & 0 & 0 & \frac{1}{2}\ea}
    \end{aligned}\end{equation}
    can be directly seen to be an optimal coupling yielding
    \begin{equation}\begin{aligned} \label{eq:D_z_upperbound0}
         D^p_{z,p}(I/2,I/2)=\Tr C_{z,p}\Pi_0=0.
    \end{aligned}\end{equation}
    \eqref{eq:D_z_upperbound-}, \eqref{eq:D_z_upperbound+}, \eqref{eq:D_z_upperbound0} together yield
    \begin{equation}\label{eq:D_z_upperbound}\begin{aligned}
        D^p_{z,p}(\rho(\alpha),\rho(\beta))\leq2^{p-1}\left(1-\sqrt{1-\max\left(\alpha^2,\beta^2\right)}\right),
    \end{aligned}\end{equation}
    without further assumptions other than $\rho(\alpha),\rho(\beta)$ having the form of \eqref{eq:rho_alpha_x}.
    Now let $M=\max(\abs{\alpha},\abs{\beta})$, suppose that $M<1$ and consider
    \begin{equation}\begin{aligned}
        X_{\pm}=2^{p-1}\lesq{\ba{cc} 1-\frac{1}{\sqrt{1-M^2}} & \pm\sqrt{\frac{M^2}{1-M^2}} \\ \pm\sqrt{\frac{M^2}{1-M^2}} & 1-\frac{1}{\sqrt{1-M^2}} \ea},\quad Y=\lesq{\ba{cc} 0 & 0 \\ 0 & 0 \ea}.
    \end{aligned}\end{equation}
    Clearly, $X_\pm$ and $Y$ are self-adjoint. It is also evident by Sylvester's criterion, that
    \begin{equation}\begin{aligned}
        C_{z,p}-Y\otimes I^T-I\otimes X_\pm^T=
        2^{p-1}\lesq{\ba{cccc} \frac{1}{\sqrt{1-M^2}}-1& \mp\sqrt{\frac{M^2}{1-M^2}} & 0 & 0\\
        \mp\sqrt{\frac{M^2}{1-M^2}} & \frac{1}{\sqrt{1-M^2}}+1 & 0 & 0\\
        0 & 0 & \frac{1}{\sqrt{1-M^2}}+1 & \mp\sqrt{\frac{M^2}{1-M^2}}\\
        0 & 0 & \mp\sqrt{\frac{M^2}{1-M^2}} & \frac{1}{\sqrt{1-M^2}}-1\ea}\geq 0,
    \end{aligned}\end{equation}
    as well as
    \begin{equation}\begin{aligned}
        C_{z,p}-X_\pm\otimes I^T-I\otimes Y^T=
        2^{p-1}\lesq{\ba{cccc} \frac{1}{\sqrt{1-M^2}}-1& 0 & \mp\sqrt{\frac{M^2}{1-M^2}} & 0\\
        0 & \frac{1}{\sqrt{1-M^2}}+1 & 0 & \mp\sqrt{\frac{M^2}{1-M^2}}\\
        \mp\sqrt{\frac{M^2}{1-M^2}} & 0 & \frac{1}{\sqrt{1-M^2}}+1 & 0\\
        0 & \mp\sqrt{\frac{M^2}{1-M^2}} & 0 & \frac{1}{\sqrt{1-M^2}}-1\ea}\geq 0.
    \end{aligned}\end{equation}
    Thus
    \begin{equation}\begin{aligned}\label{eq:D_z_lowerbound}
        &D^p_{z,p}(\rho(\alpha),\rho(\beta))\geq  \max\left(\Tr X_+\rho(\alpha),\Tr X_-\rho(\alpha),\Tr X_+\rho(\beta),\Tr X_-\rho(\beta)\right)\\
        =&2^{p-1}\left(1-\frac{1}{\sqrt{1-M^2}}\right)+2^{p-1}M\sqrt{\frac{M^2}{1-M^2}}\\
        =&2^{p-1}\left(1-\frac{1-M^2}{\sqrt{1-M^2}}\right)=2^{p-1}\left(1-\sqrt{1-M^2}\right)=2^{p-1}\left(1-\sqrt{1-\max\left(\alpha^2,\beta^2\right)}\right).
    \end{aligned}\end{equation}
    For mixed states, combining \eqref{eq:D_z_upperbound} and \eqref{eq:D_z_lowerbound} completes the proof. If either state is pure then it is known that there is only one coupling and therefore \eqref{eq:D_z_upperbound} is an equality rather than an upper bound.
\end{proof}

\begin{theorem}\label{thm:D_z_for_xy_states}
    Let $\rho$ denote the standard Bloch parametrization, that is,
    \begin{equation}\label{eq:Bloch-para}\begin{aligned}
        \rho(\vec{r}):=\frac{1}{2}\left(I+\vec{r}\cdot \vec{\sigma}\right),
    \end{aligned}\end{equation}
    and let us assume that both $\vec{r}_1$ and $\vec{r}_2$ are orthogonal to $(0,0,1).$ Then
    \begin{equation}\label{eq:D_xy_gen}\begin{aligned}
        D^p_{z,p}(\rho(\vec{r}_1),\rho(\vec{r}_2))=2^{p-1}\left(1-\sqrt{1-\max\left(r_1^2,r_2^2\right)}\right).
    \end{aligned}\end{equation}
    In particular, if $r_1\geq r_2$, then
    \begin{equation}\begin{aligned}
        D^p_{z,p}(\rho(\vec{r}_1),\rho(\vec{r}_2))=D^p_{z,p}(\rho(\vec{r}_2),\rho(\vec{r}_1))=D^p_{z,p}(\rho(\vec{r}_1),\rho(\vec{r}_1)).
    \end{aligned}\end{equation}
\end{theorem}
\begin{proof}
    This follows immediately from the invariance of \eqref{eq:Cz} under unitary conjugations implementing rotations around the $\sigma_z$ axis  (see \eqref{eq:unit-rot-1} and \eqref{eq:unit-rot-2}), and Proposition~\ref{prop:D_z_for_xy_states}.
\end{proof}
The above obtained formula for the $p$-Wasserstein distance $D_{z,p}$ gives rise to an explicit closed form for the corresponding quadratic Wasserstein divergence $d_{z,2}$ --- this is the content of the next corollary.
\begin{corollary}\label{cor:d_z_for_xy_states}
    Let $\rho$ denote the Bloch parametrization as is \eqref{eq:Bloch-para}, and let the us consider the cost operator be $C_{z,2}$ given in \eqref{eq:Cz}. Assume that both $\vec{r}_1$ and $\vec{r}_2$ are orthogonal to $(0,0,1)$. Then we have 
    \begin{equation}\begin{aligned}
        d^2_{z,2}(\rho(\vec{r}_1),\rho(\vec{r}_2))=\sqrt{1-\min(r_1,r_2)^2}-\sqrt{1-\max(r_1,r_2)^2}.
    \end{aligned}\end{equation}
\end{corollary}
\begin{proof}
    Immediate from Theorem~\ref{thm:D_z_for_xy_states}, as
    \begin{equation}\begin{aligned}
         d^2_{z,2}(\rho(\vec{r}_1),\rho(\vec{r}_2))&=D^2_{z,2}(\rho(\vec{r}_1),\rho(\vec{r}_2))-\frac{1}{2}\left(D^2_{z,2}(\rho(\vec{r}_1),\rho(\vec{r}_1))+D^2_{z,2}(\rho(\vec{r}_2),\rho(\vec{r}_2))\right)\\
         &=D^2_{z,2}(\rho(\vec{r}_+),\rho(\vec{r}_+))-\frac{1}{2}\left(D^2_{z,2}(\rho(\vec{r}_1),\rho(\vec{r}_1))+D^2_{z,2}(\rho(\vec{r}_2),\rho(\vec{r}_2))\right)\\
         &=\frac{1}{2}\left(D^2_{z,2}(\rho(\vec{r}_+),\rho(\vec{r}_+))-D^2_{z,2}(\rho(\vec{r}_-),\rho(\vec{r}_-))\right),\\
    \end{aligned}\end{equation}
    where $(\vec{r}_+,\vec{r}_-)$ is a permutation of $(\vec{r}_1,\vec{r}_2)$, so that $\abs{\vec{r}_+}\geq\abs{\vec{r}_-}$.
\end{proof}
\begin{remark}
    The quantity $\sqrt{1-r^2}$ in Theorem~\ref{thm:D_z_for_xy_states} and Corollary~\ref{cor:d_z_for_xy_states} is the length of the tangent that can be drawn from the perimeter of the circle given by the intersection of the Bloch ball and the $xy$ plane to the centered circle of radius $r$ on which the qubit $\rho(\vec{r})$ lies.
\end{remark}

A consequence of the closed formula for single observable cost and qubits perpendicular to the observable is that we can prove the quadratic triangle inequality in this case as follows.

\begin{proposition}
    Let $\rho,\sigma,\omega \in \cS(\C^2)$ be quantum bits such that all of them are orthogonal to $\sigma_z$ in the Hilbert-Schmidt sense. Then even the square of the quadratic Wasserstein divergence $d_{z,2}$ satisfies the triangle inequality, that is,
    \begin{equation}\begin{aligned}
        d^2_{z,2}(\rho,\sigma)+d^2_{z,2}(\sigma,\omega)\geq d^2_{z,2}(\rho,\omega).
    \end{aligned}\end{equation}
\end{proposition}
\begin{proof}
Let $\vec{r}_\rho$, $\vec{r}_\sigma$ and $\vec{r}_\omega$ be the Bloch vectors of $\rho,\sigma$ and $\omega$. By Corollary~\ref{cor:d_z_for_xy_states},
    \begin{equation}\label{eq:triang_ineq_xz}\begin{aligned}
        &d^2_{z,2}(\rho,\sigma)+d^2_{z,2}(\sigma,\omega)-d^2_{z,2}(\rho,\omega)=\sqrt{1-\min(r_\rho,r_\sigma)^2}-\sqrt{1-\max(r_\rho,r_\sigma)^2}\\
        +&\sqrt{1-\min(\vphantom{r_\rho}r_\sigma,r_\omega)^2}-\sqrt{1-\max(\vphantom{r_\rho}r_\sigma,r_\omega)^2}-\sqrt{1-\min(r_\rho,r_\omega)^2}+\sqrt{1-\max(r_\rho,r_\omega)^2}.
    \end{aligned}\end{equation}
If $r_\rho\leq r_\sigma\leq r_\omega$, then \eqref{eq:triang_ineq_xz} takes the following form:
\begin{equation}\begin{aligned}
        &d^2_{z,2}(\rho,\sigma)+d^2_{z,2}(\sigma,\omega)-d^2_{z,2}(\rho,\omega)=\sqrt{1-r_\rho^2}-\sqrt{1-r_\sigma^2\vphantom{r_\rho}}\\
        +&\sqrt{1-r_\sigma^2\vphantom{r_\rho}}-\sqrt{1-r_\omega^2\vphantom{r_\rho}}-\sqrt{1-r_\rho^2}+\sqrt{1-r_\omega^2\vphantom{r_\rho}}=0.
    \end{aligned}\end{equation}
If $r_\rho\leq r_\omega\leq r_\sigma$, then \eqref{eq:triang_ineq_xz} takes the following form:
\begin{equation}\begin{aligned}
        &d^2_{z,2}(\rho,\sigma)+d^2_{z,2}(\sigma,\omega)-d^2_{z,2}(\rho,\omega)=\sqrt{1-r_\rho^2}-\sqrt{1-r_\sigma^2\vphantom{r_\rho}}\\
        +&\sqrt{1-r_\omega^2\vphantom{r_\rho}}-\sqrt{1-r_\sigma^2\vphantom{r_\rho}}-\sqrt{1-r_\rho^2}+\sqrt{1-r_\omega^2\vphantom{r_\rho}}=2\left(\sqrt{1-r_\omega^2\vphantom{r_\rho}}-\sqrt{1-r_\sigma^2\vphantom{r_\rho}}\right),
    \end{aligned}\end{equation}
which is nonnegative by assumption.
If $r_\sigma\leq r_\rho\leq r_\omega$, then \eqref{eq:triang_ineq_xz} takes the following form:
\begin{equation}\begin{aligned}
        &d^2_{z,2}(\rho,\sigma)+d^2_{z,2}(\sigma,\omega)-d^2_{z,2}(\rho,\omega)=\sqrt{1-r_\sigma^2\vphantom{r_\rho}}-\sqrt{1-r_\rho^2}\\
        +&\sqrt{1-r_\sigma^2\vphantom{r_\rho}}-\sqrt{1-r_\omega^2\vphantom{r_\rho}}-\sqrt{1-r_\rho^2}+\sqrt{1-r_\omega^2\vphantom{r_\rho}}=2\left(\sqrt{1-r_\sigma^2\vphantom{r_\rho}}-\sqrt{1-r_\rho^2}\right),
    \end{aligned}\end{equation}
which is nonnegative by assumption. The other three cases of the order of $r_\rho, r_\sigma, r_\omega$ can be transformed into either one of the above using the fact that \eqref{eq:triang_ineq_xz} is symmetric in $r_\rho$ and $r_\omega$. Thus \eqref{eq:triang_ineq_xz} is nonnegative which completes the proof.
\end{proof}

We conclude this section by a simple computation which is an ingredient of the proof of Proposition \ref{prop:strictly-smaller}, and also a sanity check showing that one gets back the classical optimal transportation problem when both states involved commute with the observable generating the transport cost.

\begin{proposition}\label{prop:D_z_for_z_states}
    Let
    \begin{equation}\label{eq:rho_alpha_z}\begin{aligned}
        \rho(\alpha):=\frac{1}{2}\left(I+\alpha\sigma_z\right)
        =\frac{1}{2}\lesq{\ba{cc} 1+\alpha & 0 \\ 0 & 1-\alpha \ea},
    \end{aligned}\end{equation}
    then
    \begin{equation}\label{eq:D_z_comm}\begin{aligned}
        D^p_{z,p}(\rho(\alpha),\rho(\beta))=2^{p-1}\abs{\alpha-\beta}.
    \end{aligned}\end{equation}
\end{proposition}
\begin{proof}
    Consider
    \begin{equation}\begin{aligned}
        \Pi(\alpha,\beta):=\frac{1}{2}\lesq{\ba{cccc} 1+\min(\alpha,\beta) & 0 & 0 & 0\\
        0 & \max(\beta-\alpha,0) & 0 & 0\\
        0 & 0 & \max(\alpha-\beta,0) & 0\\
        0 & 0 & 0 & (1-\max(\alpha,\beta))\ea}.
    \end{aligned}\end{equation}
    $\Pi(\alpha,\beta)$ is clearly a coupling of $\rho(\alpha)$ and $\rho(\beta)$ and thus
    \begin{equation}\begin{aligned} \label{eq:D_z_upperboundz}
         D^p_{z,p}(\rho(\alpha),\rho(\beta))&\leq \tr \lesq{C_{z,p}\Pi(\alpha,\beta)}=2^{p-1}\abs{\alpha-\beta}.
    \end{aligned}\end{equation}
    Now consider
    \begin{equation}\begin{aligned}
        X=\lesq{\ba{cc}2^p & 0 \\ 0 & 0 \ea},\quad Y=-X=\lesq{\ba{cc} -2^{p} & 0 \\ 0 & 0 \ea}.
    \end{aligned}\end{equation}
    Clearly, $X$ and $Y$ are self-adjoint. It is also evident that
    \begin{equation}\begin{aligned}
        C_{z,p}-Y\otimes I^T-I\otimes X^T=
        \lesq{\ba{cccc} 0 & 0 & 0 & 0\\
        0 & 2^{p+1} & 0 & 0\\
        0 & 0 & 0 & 0\\
        0 & 0 & 0 & 0\ea} \geq 0,
    \end{aligned}\end{equation}
    and similarly
    \begin{equation}\begin{aligned}
        C_{z,p}-X\otimes I^T-I\otimes Y^T=
        \lesq{\ba{cccc} 0 & 0 & 0 & 0\\
        0 & 0 & 0 & 0\\
        0 & 0 & 2^{p+1} & 0\\
        0 & 0 & 0 & 0\ea}\geq 0.
    \end{aligned}\end{equation}
    It follows that
    \begin{equation}\begin{aligned} \label{eq:D_z_lowerboundz}
         D^p_{z,p}(\rho(\alpha),\rho(\beta))&\geq \max\left(\tr \lesq{ X\left(\rho(\alpha)-\rho(\beta)\right)},\Tr \lesq{X\left(\rho(\beta)-\rho(\alpha)\right)}\right)=2^{p-1}\abs{\alpha-\beta}.
    \end{aligned}\end{equation}
\end{proof}
The following is an immediate corollary.
\begin{corollary}\label{cor:d_z_for_z_states}
    Let
    \begin{equation}\begin{aligned}
        \rho(\alpha):=\frac{1}{2}\left(I+\alpha\sigma_z\right)=\frac{1}{2}\begin{pmatrix}1+\alpha & 0 \\ 0 & 1-\alpha \end{pmatrix},
    \end{aligned}\end{equation}
    then
    \begin{equation}\begin{aligned}
        d^2_{z,2}(\rho(\alpha),\rho(\beta))=D^2_{z,2}(\rho(\alpha),\rho(\beta))=2\abs{\alpha-\beta}.
    \end{aligned}\end{equation}

\end{corollary}

\paragraph*{{\bf Acknowledgment.}} We thank the anonymous referee for his/her valuable comments and insightful suggestions.


\begin{small}
\bibliographystyle{plainurl}  
\bibliography{references.bib}
\end{small}

\end{document}